\documentclass[letterpaper,11pt]{article}
\usepackage[margin=1in]{geometry}

\usepackage{booktabs}
\usepackage[ruled]{algorithm2e}

\SetAlFnt{\small}
\SetAlCapFnt{\small}
\SetAlCapNameFnt{\small}
\SetAlCapHSkip{0pt}
\IncMargin{-\parindent}

\usepackage{multirow}
\usepackage{lmodern}
\usepackage[colorlinks=true,allcolors=teal]{hyperref}
\usepackage[utf8]{inputenc}
\usepackage[T1]{fontenc}
\usepackage{microtype}
\usepackage{dsfont}
\newcommand{\bone}{\mathds{1}}
\usepackage{enumitem}
\setlist{noitemsep}
\setlist[description]{leftmargin=\parindent,labelindent=0pt}
\setlist[itemize]{leftmargin=\parindent,labelindent=0pt}

\usepackage{graphicx}
\usepackage{subcaption}
\usepackage[sort,comma,authoryear]{natbib}
\setcitestyle{square}
\usepackage{url}

\usepackage[normalem]{ulem}

\usepackage{amsmath,amsthm,amsfonts}
\usepackage{mathtools,thmtools,thm-restate}
\usepackage{xcolor}
\usepackage[capitalize,nameinlink]{cleveref}
\usepackage{booktabs}
\crefname{demproc}{Process}{Processes}
\crefname{experiment}{Experiment}{Experiments}
\newcommand{\makedemproc}{\SetAlgorithmName{Process}{demproc}{List of Democratic Processes}}

\theoremstyle{plain}
\newtheorem{theorem}{Theorem}
\crefname{theorem}{Theorem}{Theorems}
\newtheorem{lemma}[theorem]{Lemma}
\crefname{lemma}{Lemma}{Lemmas}

\crefname{conjecture}{Conjecture}{Conjectures}

\crefname{corollary}{Corollary}{Corollaries}
\newtheorem{proposition}[theorem]{Proposition}
\crefname{proposition}{Proposition}{Propositions}
\theoremstyle{definition}

\crefname{example}{Example}{Examples}
\newtheorem{definition}[theorem]{Definition}
\crefname{definition}{Definition}{Definitions}
\crefname{figure}{Figure}{Figures}

\allowdisplaybreaks

\title{Generative Social Choice}

 \author{\textbf{Sara Fish$^1$, Paul G\"olz$^2$, David C. Parkes$^1$, Ariel D. Procaccia$^1$,}\\\textbf{Gili Rusak$^1$, Itai Shapira$^1$,} and \textbf{Manuel W\"uthrich$^1$}}
 \date{\large $^1$Harvard University\hspace{.3cm}$^2$Cornell University}

\newcommand{\pg}[1]{\textcolor{black}{#1}}

\definecolor{english}{rgb}{0.0, 0.5, 0.0}

\newcommand{\disquery}[2]{\textsc{Disc}(#1, #2)} 
\newcommand{\genquery}[2]{\textsc{Gen}(#1, #2)}

\newcommand{\emdash}{\,---\,}

\newcommand{\rthlargest}[1]{\mathrm{max}_{(#1)}}

\renewcommand{\theta}{\vartheta}

\DeclareMathOperator*{\argmin}{argmin}
\DeclareMathOperator*{\argmax}{argmax}
\newcommand{\nsubseteq}{\not\subseteq}

\newenvironment{prompt}{\begingroup\ttfamily}{\endgroup}

\usepackage{soul}

\makeatletter
\newcommand{\removelatexerror}{\let\@latex@error\@gobble}
\makeatother

\begin{document}

\maketitle

\begin{abstract}
The mathematical study of voting, \emph{social choice theory}, has traditionally only been applicable to choices among a few predetermined alternatives, but not to open-ended decisions such as collectively selecting a textual statement.
We introduce \emph{generative social choice}, a design methodology for open-ended democratic processes that combines the rigor of social choice theory with the capability of large language models to generate text and extrapolate preferences.
Our framework divides the design of AI-augmented democratic processes into two components: first, proving that the process satisfies representation guarantees when given access to oracle queries; second, empirically validating that these queries can be approximately implemented using a large language model. 
We apply this framework to the problem of summarizing free-form opinions into a proportionally representative slate of opinion statements; specifically, we develop a democratic process with representation guarantees and use this process to {portray} the opinions of participants in a survey about {abortion policy}.
In a trial with 100 representative US residents, we find that {84 out of 100 participants feel ``excellently'' or ``exceptionally''} represented by the slate of five statements we extracted.
\end{abstract}

\section{Introduction}
Voting is a key way in which groups\emdash{}be they national electorates, members of a legislature, or members of a board\emdash{}make common decisions.
The theoretical foundation of voting is provided by the field of \emph{social choice theory}, which studies mathematical guarantees in the context of how different voting rules aggregate individual preferences into a collective decision.
The typical social choice setting involves a \emph{small, predetermined set of alternatives} (e.g., the candidates in an election), over which voters specify their preferences and from which the voting rule selects an outcome.

Many pressing policy questions, however, are too nuanced to fit this neat template of choosing between a few alternatives.
The need for open-ended forms of democratic input is demonstrated, for example, by the increased use of \emph{deliberative minipublics}~\citep{OECD20,FGGH+21}, which provide policy recommendations to governments on complex issues such as climate change~\citep{WCS22} and electoral reform~\citep{Fournier11}.
Similarly nuanced questions arise around the alignment of artificial intelligence~(AI) with societal interests; in this context, Meta~\citep{Clegg23} and Open\-AI~\citep{EL24} have been experimenting with democratic processes that seek public input to open-ended questions such as ``how far [\dots] personalization of AI assistants like ChatGPT to align with a user's tastes and preferences should go?''~\citep{ZDAE+23}.
Though deliberation can address such open-ended questions, it lacks two key strengths of voting: scalability~\cite[e.g.,][]{Goodin00} and guarantees on its outcomes.

To address these shortcomings, we introduce a new paradigm for the design of democratic processes: \emph{generative social choice}.
It fuses the rigor of social choice theory with the flexibility and power of generative AI, in particular large language models (LLMs), to reach collective answers to open-ended questions in a scalable and principled way.

\subsection{How LLMs Address the Limitations of Classical Social Choice}
In our view, there are two fundamental obstacles to applying classical social choice to open-ended questions, both of which can be overcome by LLMs. \smallskip

\begin{itemize}
\item \textit{Unforeseen Alternatives.}
In classical social choice, the set of alternatives is explicitly specified and static.
Take the 2016 Brexit referendum, for example, in which the alternatives were either to maintain the status quo or make a clean break with the European Union.
Since intermediate options were not specified, they could not be selected by voters, even if they might have enjoyed a much larger degree of support.
Even in participatory budgeting~\citep{Cab04}, the set of alternatives is limited to the budget-feasible subsets of \emph{previously proposed} projects.

By contrast, LLMs have the capability of \emph{generating} alternatives that were not initially anticipated but find common ground between agents.
In principle, the possible outcomes of an LLM-augmented democratic process may span the universe of all relevant outcomes for the problem at hand, e.g., all possible bills or statements. \smallskip

\item \emph{Extrapolating Preferences.}
In classical social choice theory, voters specify their preferences in a rigid format.
Typically, agents evaluate each alternative independently, or, if the alternatives form a combinatorial domain,\footnote{This is the case, for example, in multi-winner elections or participatory budgeting.} a voting rule might assume that preferences have a restricted parametric shape and only elicit its parameters.
Clearly, this approach does not suffice if a democratic process may produce alternatives that were not previously anticipated, and therefore not elicited: to even know which alternatives would be promising to generate, the process must be able to extrapolate agents' preferences.

LLMs can address this problem as they enable participants to \emph{implicitly} specify their preferences by expressing their opinions, values, or criteria in natural language. The LLM can act as a proxy for the participant, predicting their preferences over any alternative, whether foreseen or newly generated.
\end{itemize}

\subsection{A Framework for Generative Social Choice}
It is clear, at this point, what LLMs can contribute to social choice.
LLMs and social choice theory make an odd couple, however, because social choice focuses on rigorous guarantees whereas LLMs are notoriously impervious to theoretical analysis.
We propose a framework for generative social choice that addresses this difficulty by breaking the design of democratic processes into two interacting components.

\begin{itemize}
\item \emph{First component: Guarantees with perfect queries.}
Assume that the LLM is an oracle that can precisely answer certain types of queries, which may involve generating new alternatives in an optimal way or predicting agents' preferences.
Once appropriate queries have been identified, the task is to design algorithms that, when given access to an oracle for these queries, provide social choice guarantees.

\item \emph{Second component: Empirical validation of queries.}
Assuming the LLM to be a perfect oracle is helpful for guiding the design of a democratic process, but of course not an accurate reflection of reality.
In the second component, the task is to implement the proposed queries using calls to an LLM, and to empirically validate how well these implementations match the queries. 
\end{itemize}
Naturally, the two components interact:
The theory identifies queries that are useful for social choice and should hence be validated empirically.
Conversely, experiments show which queries can be answered accurately in practice, raising the question of which guarantees algorithms relying on these queries might provide.

A key benefit of this framework is that theoretical results derived in it are future-proof: as LLMs continue to rapidly improve, they will only grow more reliable in answering queries, making the LLM-based aggregation methods ever more powerful.

\subsection{Our Results: A Case Study in Generative Social Choice}
In addition to introducing the framework presented above, we demonstrate it in one particular setting: summarizing a large body of free-form opinions into a slate of few statements, in a representative manner.
In this setting, participants share free-form opinions about a given policy issue on an online platform such as \emph{Polis}~\citep{SBES+21} or \emph{Remesh},\footnote{\url{https://www.remesh.ai/}} or as part of a qualitative survey.
Then, a voting rule selects a slate of $k$ statements that is  proportionally representative of the diversity and relative prevalence of viewpoints among the participant population.

The setting of statement selection was formalized by~\citet{HKPT+23} in the language of multi-winner approval elections:
If we think of statements as candidates, and of an agreement between participants and statements as binary approval votes, the slate should satisfy axioms for proportional representation from this literature such as \emph{justified representation}~(JR).
In our work, we allow cardinal (rather than just binary) levels of participant--statement agreement.
Furthermore, we introduce a novel strengthening of JR, \emph{balanced justified representation (BJR)}, which we believe to be particularly well suited for our statement-selection application and of independent interest.

Whereas previous summarization systems can only select a slate among users' statements, our process
can \emph{generate new statements}, which might find new common ground between participants and allow for more representative slates. Our process takes as input each user's interactions on the platform as a description of their preferences.
The process then employs an LLM to \emph{(1)}~translate these descriptions into participants' utilities for any new statements (\emph{discriminative} queries, in the language of machine learning), and \emph{(2)}~generate statements that maximize the utility of a subset of participants, based on their descriptions (\emph{generative} queries).

Following our framework's first component, we show that, with access to polynomially many of these queries, a democratic process resembling \emph{Greedy Approval Voting}~\citep{ABCE+17} guarantees BJR.
Crucially, this guarantee holds not just relative to a set of predetermined statements but to the space of all possible statements (\cref{sec:theory:algs}).

A potential issue with this process is that, through the generative query, it calls the LLM with a prompt whose length scales linearly in the number of participants.
This is problematic since LLMs can only handle input of bounded length.
We show that, unless one makes assumptions on the structure of preferences, this problem of linear-size queries is unavoidable for any process guaranteeing BJR with subexponentially many queries (\cref{sec:theory:tqueries}).
If, however, the space of statements and preferences is structured, specifically, if it has finite VC dimension \citep{Vapnik1998-fp}, democratic processes based on sampling can guarantee BJR (with high probability) using a polynomial number of queries whose length is independent of the number of participants (\cref{sec:theory:structured}).

In \cref{sec:experiments}, we present a practical, LLM-based implementation of discriminative and generative queries. 
Empirical validation shows that the proposed implementation of the discriminative query accurately extrapolates agents' preferences to unseen statements. Further, we show that the proposed implementation of the generative query consistently produces high-agreement statements.

Equipped with these query implementations, we then deploy the full democratic process in \cref{sec:pilot}.
{We pilot our process to study US residents' opinions on abortion policy.}
We elicit free-text opinions about this topic from a sample of 100 participants and distill them into a representative slate of five statements, using our LLM-enhanced democratic process.
To validate that these statements faithfully represent the population, we conduct a second survey with a fresh sample of 100 US residents.
{After matching the five statements to equal-sized blocs of participants,  
84\% of participants say that their assigned statement captures their view on abortion policy ``excellently'' or ``exceptionally''. Only a single participant feels less than ``well'' represented (the midpoint of our scale) by their assigned statement.} To support future research on online participation, we made the participants' full responses available as a public data set.\footnote{\url{https://github.com/generative-social-choice/survey_data/}} 

\subsection{Related Work}
In a position paper that is independent of our work, \citet{SVDH+23} discuss the opportunities and risks of LLMs in the context of Polis. The opportunities they identify include topic modeling, summarization, moderation, comment routing, identifying consensus, and vote prediction. Most relevant to us are their experiments for the vote prediction task, which are closely related to our implementation and evaluation of discriminative queries.
In the future, our democratic process as a whole could serve in the summarization role envisioned by \citet{SVDH+23}, for which they do not propose specific algorithms and perform no systematic experiments.

Our discriminative queries are related to the paradigm of \emph{virtual democracy}, which facilitates automated decisions on ethical dilemmas by learning the preferences of stakeholders and, at runtime, predicting their preferences over the current alternatives and aggregating the predicted preferences; example papers, which employ classical machine-learning algorithms, apply the paradigm to domains such as autonomous vehicles~\citep{NGAD+18}, food rescue~\citep{LKKK+19}, and kidney exchange~\citep{FSSD+20}. These papers all aim to predict preferences on a fixed set of alternatives\emdash they do not generate new alternatives.

A source of inspiration for our work is the paper of \citet{BCST+22}. They fine-tune an LLM to generate a single consensus statement for a specific group of people, based on written opinions and ratings of candidate statements. Reward models are trained to capture individual preferences, and the acceptability of a statement for the group is measured through a social welfare function. More recent work by \citet{TBJS+24}, which is concurrent with ours, also seeks to generate consensus statements by modeling rewards and fine-tuning a large language model; social choice plays a direct role, as rankings over candidate statements are aggregated using a voting rule due to \citet{Schul03}. Among other findings, \citet{TBJS+24} show that the consensus statements produced by their system are preferred by participants to those proposed by mediators. One difference between these two papers and ours is that we do not attempt to find a single statement that builds consensus across the entire group\emdash we instead allow for multiple statements representing distinct opinions.
Despite this difference, part of our evaluation of the generative query adapts an experimental setup by \citet{BCST+22}.
A more fundamental difference is that we view our experiments as an instance of a broader framework that allows for a systematic investigation of the types of queries an LLM can perform and the theoretical guarantees they provide.

Finally, we build
on the rich literature on justified representation in approval-based committee elections~\citep{ABCE+17,LS23a}.
As we have already mentioned, \citet{HKPT+23} also study representation axioms from this literature in a statement-selection context.
The key technical challenge in their work is that they only have access to partial approval votes.
The learning-theoretic approach they adopt, as well as a later refinement by \citet{BP23}, bears technical similarity to the algorithm we propose for obtaining representation with size-constrained generative queries.
All previous papers in this literature assume a non-generative setting with a fixed set of alternatives.

\section{Model}

Let $N$ be a set of $n$ \emph{agents}, and let $\mathcal{U}$ denote the \emph{universe of} (well-formed, on-topic) \emph{statements}, which may be finite or infinite. Each agent $i\in N$ has a \emph{utility function} $u_{i}:\mathcal{U}\to\mathbb{R}$ that maps statements to utilities.
Whereas our positive results apply for arbitrary real-valued utility functions, our impossibilities will even hold in the restricted setting of \emph{approval utilities}, where utilities are 0 or 1, which much of the prior work has focused on \citep{LS23a}.
An \emph{instance} of the statement-selection problem consists of $N$, $\mathcal{U}$, $\{u_{i}\}_{i\in N}$, and a \emph{slate size} $k\in\mathbb{N}_{\geq1}$.

A \emph{democratic process} is an algorithm that, when run on an instance, returns a \emph{slate}, i.e., a multiset consisting of $k$ statements from the universe.\footnote{Allowing a slate to contain the same statements multiple times avoids technical problems with the edge case where generative queries return the same statement, in which case no query-based algorithm would be able to procure $k$ distinct statements. We also believe this choice to be suitable for our application domain, where representing multiple segments of the population by identical statements might sometimes be appropriate, for example if all agents in these segments have identical preferences. For ease of exposition, we will slightly abuse notation and treat slates as if they were sets; this essentially amounts to assuming that different generative queries do not return exactly the same statement.}
Crucially, this algorithm receives only $N$ and $k$ in its input, but not $\mathcal{U}$ or the $u_{i}$, which it must instead access through queries as we describe below.
{Note that our model treats the desired slate size $k$ as fixed. In settings where a range of slate sizes is permissible, one can use the democratic process to produce slates of all valid sizes and then select the slate maximizing a measure of fit, analogously to determining the number of clusters in clustering~\cite[][p.\ 126 ff.]{ELL+11}.}

For convenience, we denote the $r$th largest element in a finite set $X$ of real numbers (for $1 \leq r \leq |X|$) by $\rthlargest{r}(X)$.
To deal with edge cases, we set $\rthlargest{0}(X) \coloneqq \infty$ for all sets $X$.

\subsection{Queries}

Since the democratic process does not receive the statements and preferences in its input, it instead accesses them indirectly through \emph{queries}. The democratic processes we develop make use of two query types: 
\begin{description}
\item [Discriminative Queries.]
Discriminative queries extrapolate an agent's utility function to unseen statements. For an agent $i$ and statement $\alpha$, $\disquery{i}{\alpha}$ returns $u_{i}(\alpha)$. 
\item [Generative Queries.]  
For a set of agents $S$ of size at most $t$ and an integer $0 \leq r \leq |S|$, $t$-$\genquery{S}{r}$ returns the statement in $\mathcal{U}$ that maximizes the $r$-highest utility among the members of $S$. Formally, the query returns
\begin{equation}
\argmax_{\alpha \in \mathcal{U}} \rthlargest{r}\big(\{ u_i(\alpha) \mid i \in S \}\big),
\label{eq:generativedef}
\end{equation}
breaking ties arbitrarily.
\end{description}
Intuitively, the generative query's parameter $r$ interpolates between finding a lowest common denominator ($t$-$\genquery{S}{|S|}$ maximizes the minimum utility over $S$) and finding a statement that precisely matches a narrow coalition in $S$ (e.g., $t$-$\genquery{S}{1}$ gives some agent maximum utility, but may be unpopular among the remaining agents).
For convenience, we will simply write $\genquery{\cdot}{\cdot}$ to refer to generative queries without a size limit 
or to talk generally about generative queries with different size constraints $t$.

\subsection{Representation Axiom}

The aim of our democratic processes is to produce a slate of statements $W$ that is representative of the agent population.
If agents have approval utilities, statement selection reduces to the classic setting of multi-winner approval voting.
Therefore, our target axiom is inspired by the family of \emph{justified representation} axioms~\citep{ABCE+17} in this literature{.}

{Throughout this work, we think of the statements in $\mathcal{U}$ as expressing an entire viewpoint on the topic of discussion, not just some aspect of such a viewpoint, and of utilities as expressing how accurately and completely the statement captures the agent's viewpoint.
Our notion of representation expresses that agents have a claim to being represented to a high degree by \emph{one} statement, rather than to some average agreement across all statements on the slate\emdash a notion of representation that echoes ideas like the \emph{fully proportional representation} of \citet{Mon95} or the \emph{perfect representation} of \citet{SEL+17}.\footnote{{Other choices would have been reasonable, say, letting the statements refer to aspects of viewpoints and defining representation via a notion in which participants' claims to representation can be spread across several statements.
In our opinion-summarization setting, however, this setup would lead to slates of vague, inoffensive statements (see \cref{app:incomparablebjr}), which are of limited use for understanding the distribution of opinions in the agent population.}}
These decisions translate to the following new representation axiom:}

\begin{definition}
A slate $W$ satisfies \emph{balanced justified representation (BJR)} if there is a function $\omega: N\to W$, matching agents {to statements such that each statement on the slate is matched to $\lfloor n/k \rfloor$ or $\lceil n/k\rceil$ agents}, for which there is no coalition $S\subseteq N$, statement $\alpha\in\mathcal{U}$, and threshold $\theta \in\mathbb{R}$ such that \emph{(i)}~$|S|\geq n/k$, \emph{(ii)}~$u_{i}(\alpha)\ge \theta$ for all $i\in S$, and \emph{(iii)}~$u_{i}(\omega(i))<\theta$ for all $i\in S$.\footnote{This axiom can also be defined in a setting where slates are sets of statements, rather than multisets. In this case, the statements $\alpha$ are restricted to lie in $\mathcal{U} \setminus W$, to make the axiom satisfiable. This axiom can be satisfied by a variant of \cref{alg:cjr}, in which the choice of statements in each iteration is restricted to statements that have not previously been selected.\label{fn:axiomforsets}}
\end{definition}

In words, if there is a coalition of agents that is \emph{(i)} large enough to ``deserve'' a statement on the slate by proportionality and \emph{(ii)} has cohesive preferences (i.e., there is a statement for which all these agents have utility at least $\theta$), then \emph{(iii)} the coalition must not be ``ignored'', in the sense that at least one member must be assigned to a statement with utility at least $\theta$.\footnote{{One might hope to strengthen this definition, so that it rules out any deviation in which the coalition members strictly increase their utility, rather than just those deviations where the coalition members all cross a common threshold $\theta$. As we show in \cref{app:incomparablebjr}, such a strengthening of BJR can be impossible to satisfy.}}

Our notion of BJR strengthens the classical axiom of justified representation, and is logically incomparable to several other axioms in the social choice literature.
We prove these relationships and {justify} the need for a new axiom in \cref{app:incomparablebjr}.
Throughout this paper, we will aim to build democratic processes that satisfy BJR, even when the universe of statements is very large and can only be navigated through queries.

\section{First Component: Guarantees with Perfect Queries}

\label{sec:theory}

In this section, we instantiate the first component of the generative social choice framework.
We defer all proofs to \cref{app:proofs}.

\subsection{Unconstrained Queries}
\label{sec:theory:algs}

We begin by constructing a democratic process that guarantees BJR in polynomial time. This algorithm uses queries of type $\disquery{\cdot}{\cdot}$ and $n$-$\genquery{\cdot}{\cdot}$, i.e., generative queries without constraints on the number of input agents.
The democratic process we propose, shown in \cref{alg:cjr}, can either be seen as a generalization of \emph{Greedy Approval Voting}~\citep{ABCE+17}, or as a variant of the \emph{Greedy Monroe Rule}~\citep{SFS15} that selects statements based on an egalitarian~{\citep{BSU13}} rather than utilitarian criterion.
Our democratic process iteratively constructs a slate, adding statements one at a time.
In each iteration, it identifies a set $T$ of $n/k$ (up to rounding) remaining agents and a statement $\alpha$ such that $\min_{i \in T} u_i(\alpha)$ is maximized.
It then adds $\alpha$ to the slate, removes the agents $T$ (who are now satisfied), and repeats.
Our proof in \cref{app:proofs} that this process satisfies BJR follows in structure the argument by \citet{ABCE+17} that Greedy Approval Voting satisfies JR.

\begin{algorithm}
\DontPrintSemicolon \makedemproc \label[demproc]{alg:cjr}

\textbf{Inputs}: agents $N$, slate size $k$\;

$\bar{r} \gets n\frac{1}{k}$\;

$S\gets N$\;

$W\gets\emptyset$\;

\For{$j = 1,2,\dots,k$}{

$\alpha \gets \genquery{S}{\lceil \bar{r} \rceil}$

$W\gets W\cup\{\alpha\}$\;

$r\gets \pg{\begin{cases}
    \lceil \bar{r} \rceil & \text{if $j \leq n - k \cdot \lfloor \bar{r} \rfloor$} \\
    \lfloor \bar{r} \rfloor & \text{else} \end{cases}}$

$T \gets$ the $r$ agents in $S$ with largest $\disquery{\cdot}{\alpha}$

$S\gets S\setminus T$\; } 

\Return $W$\; \caption{Democratic Process for Balanced Justified Representation}
\end{algorithm}

\begin{restatable}{theorem}{propcjr}
\label{prop:cjr} \Cref{alg:cjr} satisfies balanced justified representation in polynomial time in $n$ and $k$, using queries of types $n$-$\genquery{\cdot}{\cdot}$ and $\disquery{\cdot}{\cdot}$. \end{restatable}

\subsection{Size-Constrained Generative Queries}
\label{sec:theory:tqueries}
So far, our generative queries could generate optimal statements even if the queried set $S$ of agents was as large as $n$.
When implementing a generative query using an LLM, however, the prompt to the LLM must include, for each agent in $S$, enough information to extrapolate the agent's preferences across the universe of statements. 
Since this information can easily take hundreds of tokens (if not more) per agent in $S$, but LLMs have bounded context windows (e.g.~128,000 tokens for GPT-4o), there are barriers to scaling to hundreds of agents and beyond. Moreover, LLMs with long context windows may struggle to effectively use the entirety of their context window \citep{LLH+23}. This motivates the design of democratic processes that function using only generative queries with bounded input size $|S|$.
Therefore, we investigate in this section whether democratic processes can still ensure BJR when generative queries are limited to sets of agents of some size $t$ that is substantially smaller than $n$. 
Immediately, we see that, if the query size $t$ is even just slightly smaller than $n/k$, representation cannot be attained:
\begin{restatable}{proposition}{propjrlower}
    \label{prop:jrlower}
    No democratic process can guarantee balanced justified representation with arbitrarily many $\frac{n}{k}\, (1 - \frac{1}{k})$-$\genquery{\cdot}{\cdot}$ and $\disquery{\cdot}{\cdot}$ queries.
    This impossibility even holds in the subsetting of approval utilities and for the weaker axiom of justified representation.
\end{restatable}
Conceptually, the proof of this theorem and the subsequent impossibility theorem are based on the idea of \emph{overshadowing}.
Specifically, we construct instances that have few ``popular'' statements and many ``unpopular'' statements with lower support.
For a given set $S$ of at most $t$ agents, our instances will ensure that some unpopular statement will be at least as well liked \emph{within $S$} as any popular statement.
Thus, all generative queries might return unpopular statements, and we design the instance such that no slate composed entirely of unpopular statements is representative.
In \cref{app:proofs}, we apply this idea in a straightforward way to prove \cref{prop:jrlower}.

On the face of it, slightly larger size-constrained generative queries seem promising for achieving BJR, since there is a democratic process that achieves BJR with queries of size $t = \lceil n/k \rceil$.
Indeed, observe that, for any $S$ and $r$,
\begin{align*}
\genquery{S}{r} &= \argmax_{\alpha \in \mathcal{U}} \rthlargest{r}\big(\{ u_i(\alpha) \mid i \in S \}\big) = \argmax_{\alpha \in \mathcal{U}} \max_{\substack{S' \subseteq S\\|S'| = r}} \rthlargest{r}\big(\{ u_i(\alpha) \mid i \in S' \}\big) \\
&= \argmax_{\alpha \in \big\{\genquery{S'}{r} \,\big|\, S' \subseteq S, |S'| = r\big\}} \rthlargest{r}\big(\{ u_i(\alpha) \mid i \in S\}\big),
\end{align*}
which shows that any call to $\genquery{S}{r}$ can be simulated by (exponentially many) $r$-$\genquery{\cdot}{\cdot}$ queries and discriminative queries.
By applying this simulation to \cref{alg:cjr}, in which all generative queries satisfy $r \leq \lceil n/k \rceil$, \cref{prop:cjr} immediately implies that BJR can be implemented by $\lceil n/k \rceil$-$\genquery{\cdot}{\cdot}$ queries, though the time complexity of the modified process is obviously prohibitive.

\begin{restatable}{proposition}{thmjrupper}
    \label{thm:jrupper}
    There exists a democratic process that satisfies balanced justified representation using (exponentially many) queries of type $\lceil n/k \rceil$-$\genquery{\cdot}{\cdot}$ and $\disquery{\cdot}{\cdot}$.
\end{restatable}

Unfortunately, the exponential running time of this na\"ive democratic process turns out to be unavoidable, even if the generative queries can have linear size in $n$.
Our proof must necessarily be more complicated than our previous impossibility in \cref{prop:jrlower}, in which we constructed an explicit instance on which \emph{any} democratic process with $t$-bounded generative queries had to violate representation.
A more sophisticated proof is necessary since, for any instance, there exists a democratic process that satisfies BJR in polynomial time with $\lceil n/k \rceil$-$\genquery{\cdot}{\cdot}$ queries on this instance; namely, a variant of the algorithm from \cref{thm:jrupper} that guesses the right subset $S'$ and returns the corresponding statement {$\genquery{S'}{r}$}.
We prove our impossibility (in \cref{app:proofs}) by showing that, for any fixed polynomial-time algorithm, \emph{there exists} an instance on which this algorithm violates BJR, through an application of the probabilistic method.
\begin{restatable}{theorem}{thmjrlowerpolytime}
    \label{thm:jrlowerpolytime}
    No democratic process can guarantee balanced justified representation with any number of $\disquery{\cdot}{\cdot}$ queries and fewer than $\frac{2}{k} \, e^{n/(12 k)}$ queries of type $\frac{n}{8}$-$\genquery{\cdot}{\cdot}$.
    This holds even for the subsetting of approval utilities and the weaker axiom of justified representation.
    As a corollary, if $k \in O(n^{0.99})$, then any democratic process guaranteeing BJR with $\frac{n}{8}$-$\genquery{\cdot}{\cdot}$ and $\disquery{\cdot}{\cdot}$ queries has exponential running time.
\end{restatable}

\subsection{Structured Preference Settings}
\label{sec:theory:structured}

While the last section's lower bounds are potentially worrisome, a silver lining is that the instances we used to prove them were contrived. Our impossibility proofs were constructed by drowning popular statements in an overwhelming number of relatively unpopular statements: for any set of agents (of a given size), there was a statement that was well liked by only these agents and not by any other agent.
Since statements and preferences in the real world presumably have some structure, it seems highly implausible that such an abundance of orthogonal statements would exist for real-world populations.
Note that, by ``structure'' we are not referring to any fixed geometry of alternatives (in contrast to, say, spatial models of voting).
Instead, we only require that preferences do not have infinite ``complexity''.

To formally define this complexity, we introduce the notion of a \emph{statement space} $(\mathcal{U},\mathcal{F})$ which consists of a universe of statements $\mathcal{U}$ and a set of possible utility functions $\mathcal{F} \subseteq \mathbb{R}^{\mathcal{U}}$.
A statement-selection instance belongs to $(\mathcal{U},\mathcal{F})$ if its universe of statements is $\mathcal{U}$ and if each agent $i$'s utility function $u_{i}$ appears in $\mathcal{F}$.

To measure the complexity of a statement space, we borrow a fundamental complexity notion from learning theory, the \emph{VC dimension} \citep{Vapnik1998-fp}. We extend the definition of VC dimension to statement spaces in a natural way: 
The VC dimension of $(\mathcal{U},\mathcal{F})$ is the largest $d\in\mathbb{N}$, for which there exist $u_{1},u_{2},\dots,u_{d}\in\mathcal{F}$ such that, for any index set $\mathcal{I}\subseteq\{1,\dots,d\}$, there is a statement $\alpha\in\mathcal{U}$ and threshold $\theta\in\mathbb{R}$ such that $u_{i}(\alpha)\ge\theta$ for all $ i\in\mathcal{I}$ and $u_{i}(\alpha)<\theta$ for all $i\notin\mathcal{I}$. If no largest integer $d$ exists, the VC dimension is infinite. 
In other words, $d$ is the size of the largest set of participants, such that for any subset of participants there is a statement that has a utility above some threshold {\emph{only} for participants within} this subset.

This notion of VC dimension of a statement space $(\mathcal{U},\mathcal{F})$ is identical to the classic, learning-theoretic VC dimension of a hypothesis set $\mathcal{H}$, constructed as follows. We define a family of functions $h_{\alpha, \theta}$ that map the utility functions $u \in \mathcal{F}$ to binary labels as follows:
\[
h_{\alpha,\theta}(u) \coloneqq \begin{cases}
1 & \text{if }u(\alpha)\ge\theta\\
0 & \text{else}
\end{cases}
\]
That is, $h_{\alpha, \theta}(u)$ indicates whether an agent with utility function $u$ assigns a utility of $\theta$ or larger to a statement $\alpha$.
Then, the VC dimension $d$ of a statement space $(\mathcal{U},\mathcal{F})$ is identical to the classic, learning-theoretic VC dimension of the hypothesis set $\mathcal{H} \coloneqq \left\{ h_{\alpha,\theta} \,\middle|\, \alpha\in\mathcal{U},\theta\in\mathbb{R}\right\}$, consisting
of binary classifiers over $\mathcal{F}$.

It seems unlikely that $d$ would be huge in real-world settings,
as it would imply, for instance (assuming a one-dimensional simplification), that we could find a statement such that people that lie at
opposite sides of the space of opinions all support that statement, while people that lie in
the middle disagree with it.
If, hence, the VC dimension of the statement space is finite in realistic settings, we can obtain BJR even with size-constrained generative queries, as formalized by the following theorem.

\begin{restatable}{theorem}{thmvccjr} \label{thm:vccjr} Let $d$ be the VC dimension of the statement space and $\delta>0$ the maximum admissible error probability. Then, \cref{alg:cjrsampling} runs in polynomial time in $n,k$ (independent of $d$) and satisfies BJR with probability at least $1-\delta$ using $\disquery{\cdot}{\cdot}$ and $t$-$\genquery{\cdot}{\cdot}$ queries for $t \in O\big(k^{4}(d\!+\!\log\frac{k}{\delta})\big)$. \end{restatable} 
The proof of this theorem can be found in \cref{app:proofs}.
The process that achieves this result, \cref{alg:cjrsampling}, is an adaptation of \cref{alg:cjr}. The key difference (\cref{ln:def2X,ln:def2Y,ln:def2alpha}) is that here we run $\genquery{Y}{\cdot}$ on a random subset $Y\subseteq N$ of the agents. Importantly, the size of this subset does not grow with the total number of agents $n$.

\begin{algorithm}[t]
\definecolor{highlightColor}{rgb}{0.156, 0.208, 0.576}

\DontPrintSemicolon \makedemproc \label[demproc]{alg:cjrsampling}

\textbf{Inputs}: agents $N$, slate size $k$, \textcolor{highlightColor}{VC dimension $d$, error probability $\delta$}\;

{\color{highlightColor}
    $n_{x}\gets O\left(k^{4}\,(d+\log(k/\delta))\right)$\;

$\epsilon\leftarrow\frac{1}{4k^{2}}$

    $\bar{r}_{x}\gets n_{x}\left(\frac{1}{k}-\epsilon\right)$\;

    $\bar{r}\gets n\left(\frac{1}{k}-2\epsilon\right)$\;
}
$S\gets N$\;

$W\gets\emptyset$\;

    \For{$j=1,2,\dots,k$}{

    {\color{highlightColor}
$X\gets$ draw $n_{x}$ agents from $N$ without replacement\;\label{ln:def2X}
    
    $Y\gets X\cap S$\;\label{ln:def2Y}

$\alpha\leftarrow\pg{\color{highlightColor}\begin{cases}
    \genquery{Y}{\left\lceil \bar{r}_{x}\right\rceil } & \text{if \ensuremath{|Y|\ge \bar{r}_{x}}}\\
\text{some arbitrary \ensuremath{\alpha\in\mathcal{U}}} & \text{else}
\end{cases}}$\;\label{ln:def2alpha}
    }

$W\gets W\cup\{\alpha\}$\;

$r\gets\pg{\begin{cases}
    \left\lceil \bar{r}\right\rceil  & \text{if \ensuremath{j\leq n-k\left\lfloor \bar{r}\right\rfloor }}\\
    \left\lfloor \bar{r}\right\rfloor  & \text{else}
\end{cases}}$

$T\gets$ the $r$ agents in $S$ with largest $\disquery{\cdot}{\alpha}$

$S\gets S\setminus T$\; }

    \Return $W$\; \caption{Democratic Process for BJR with Size-Constrained Queries. \\(differences with \cref{alg:cjr} are highlighted in color)}
\end{algorithm}

To illustrate the power of this theorem with a simple example, suppose that opinions on a discussion topic vary along three dimensions, say socially conservative vs.\ liberal, fiscally conservative vs.\ liberal, and religious vs.\ secular.
Suppose furthermore that agents and statements can be represented as points in this three-dimensional space, such that the utility $u_i(\alpha)$ is a (strictly monotonically decreasing) function of the Euclidean distance between the agent $i$ and statement $\alpha$ in this space.
Then, the hypothesis set $\mathcal{H}$ (as introduced above) of this statement space is just the set of all spheres in $\mathbb{R}^3$, which is well known to have VC dimension $d=4$~\cite[e.g.,][p.\ 122]{BHK20}.
Hence, \cref{alg:cjrsampling} produces BJR slates (up to a failure probability below $10^{-6}$) using $t$-$\genquery{\cdot}{\cdot}$ queries with $t \le \text{const} \cdot  k^4 \, (4 + \log (k) + 6\cdot\log(10))$.
If $n$ is large, this $t$ is much smaller than the lower bounds on $t$ that are implied by \cref{prop:jrlower,thm:jrlowerpolytime} when we assume an unstructured statement space.

Importantly, \cref{thm:vccjr} extends to far more complicated preference structures, and it does not require the structure to be known, but only (an upper bound on) the VC dimension. If, for example, the set of statements $\mathcal{U}$ consists of all sequences of $w$ many words in English (which has below $10^6$ words), a naive upper bound on the VC dimension of the statement space is $d \leq \log_2(|\mathcal{U}|) \leq w \, \log_2(10^6)$.
Thus, $t \le \text{const} \cdot  k^4 \, (w + \log (k))$ suffices to virtually guarantee BJR.

In summary, despite the negative worst-case results from \Cref{sec:theory:tqueries}, it is highly likely that relevant statement spaces in reality have enough structure to allow for a BJR guarantee with high probability and a relatively small number of queries, which is \emph{independent of the number of agents $n$}.
This means that we can scale the democratic process to any number of participants, say to a national audience, even when using an LLM with bounded context window size.

\section{Second Component: Empirical Validation of Queries}\label{sec:experiments}
We established in the previous section that, with access to \emph{perfect} generative and discriminative queries, we can guarantee BJR. 
In this section, we describe how we implement these queries as subprocedures interfacing with an LLM, and we empirically study how well our implementations approximate the idealized queries.

\paragraph{Evaluation Data.}
{
To evaluate the query implementations, we use the data collected in our pilot studying public opinion on abortion, which we discuss in detail in \cref{sec:pilot} and \cref{app:experiments}.
The dataset consists of survey responses by a sample of 100 US residents.
Each participant extensively describes their views on abortion in free-form responses to multiple questions, and summarizes their opinion in a self-contained statement.
Furthermore, each participant rates five example statements, each of which expresses a distinct position on how society should approach abortion in three sentences.\footnote{{These statements were generated using GPT-4o, instructed to generate views of US residents, and to make them broadly appealing. The prompt and statements can be found in \cref{app:one-shot}.}}
We elicited these ratings by asking participants ``how well does this summary capture your viewpoint on abortion?'' Participants were then asked to choose a rating on a 7-point scale with the levels ``very poorly'' (0), ``poorly'' (1), ``moderately'' (2), ``well'' (3), ``very well'' (4), ``excellently'' (5), and ``exceptionally'' (6).\footnote{{
We adopt this unbalanced Likert scale to encourage participants to be more discerning, allowing room for distinctions between statements they generally agree with and those that align with their views in greater detail and depth.
}}
We also asked participants to explain their rating in free text, by mentioning parts of the statement they agreed with, disagreed with, or parts that could be added or made more concrete.
We equate ratings with utilities; e.g., an agent $i$ rating a statement $\alpha$ with ``moderately'' means that $u_i(\alpha)=2$.}
Note, however, that the choice of numerical values is largely inconsequential given that \cref{alg:cjr} is invariant to monotone transformations of the rating scale.

\subsection{Discriminative Queries}
Our implementation of the discriminative query $\disquery{i}{\alpha}$ takes as input agent $i$'s  survey responses and the statement $\alpha$, and returns a prediction of the rating $u_i(\alpha)$.
{We implement these queries with a single call to OpenAI's GPT-4o model, as follows.
In the system prompt, we instruct the model to act as a text completion tool for a (hypothetical) user filling out an opinion survey, and ask the model to predict the user's most likely next response.
The prompt itself first lists the free-form questions along with participant $i$'s answers.
Then, it lists the rating questions, along with participant $i$'s rating and explanation.
Finally, we add one more rating question for the statement $\alpha$, and ask the model to predict the user's numeric rating level between 0 and 6. See \cref{app:discriminative_query} for details.
}

{This prompt design strikes us as promising in several ways.
The participant's free-text responses give the model in-depth information about their general thinking about abortion.
The participant's ratings of other statements serve as few-shot examples, which can help calibrate the model's prediction with the user's usage of the seven-point scale.
Finally, since the LLM's response consists of a single token (i.e., an integer between 0 and 6), we can interpret the model's token probabilities as a probability distribution capturing the model's uncertainty about $u_i(\alpha)$.
In our implementation of $\disquery{i}{\alpha}$, we return the expected $u_i(\alpha)$\emdash a real number between 0 and 6\emdash which will be used in our algorithms.
An important advantage of using the expected value (rating rather than, say, the mode) is that this reduces the prevalence of ties when \cref{alg:cjr} chooses which agents to remove from consideration.
}

\begin{figure}[htb]
\centering
\includegraphics[width=.5\textwidth]{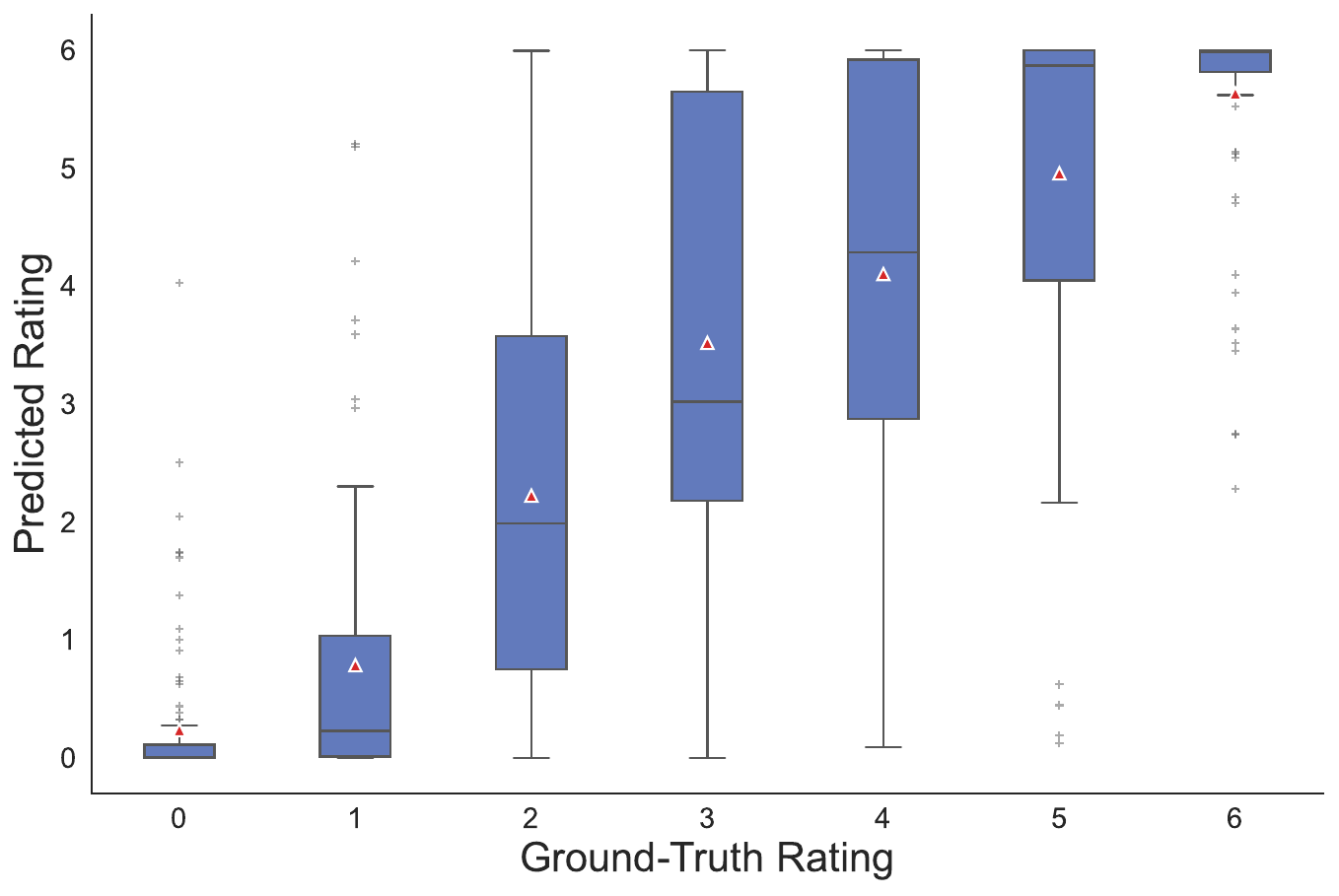}
\caption{{Distribution of discriminative query predictions on the five held-out statements for each of the 100 participants. The x-axis shows the rating level selected by the participant, the y-axis the distribution of predictions. Means are represented by triangles.}}
\label{fig:boxplot_disc_generation}
\end{figure}
\begin{figure}[htb]
\centering
\includegraphics[width=.55\textwidth]{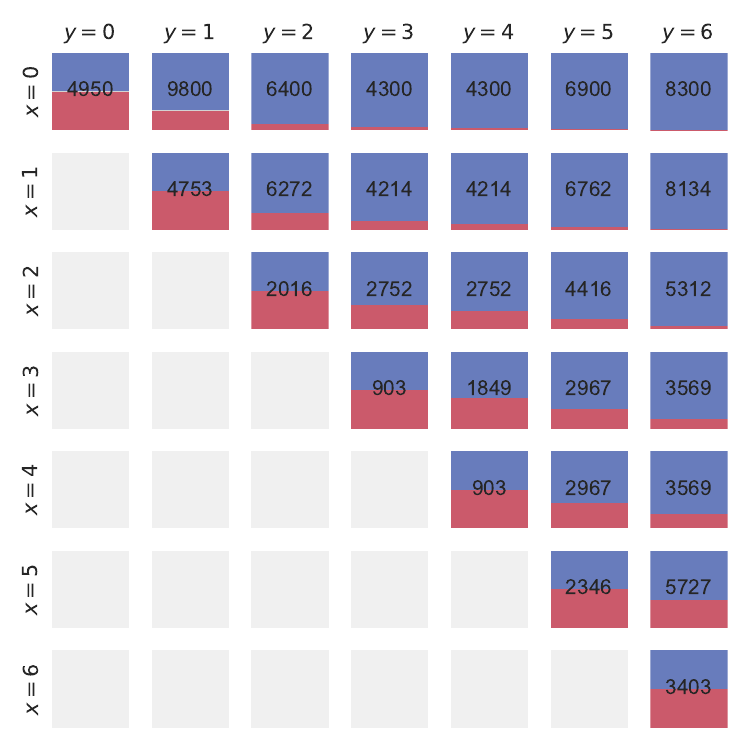}

\begin{minipage}{\textwidth}
    \centering
    \includegraphics[width=0.7\textwidth]{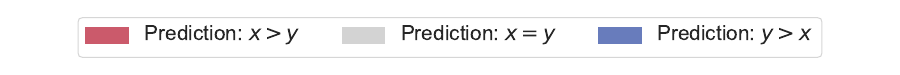}
\end{minipage}

\caption{{Predicted preference ordering for pairs of participant and held-out statement. Cell $(x, y)$ aggregates all pairs where the first participant rates the first statement at level $x$ and the second participant rates the second statement at level $y$. Colorful areas indicate in what fraction the discriminative query for the first pair gives a higher, equal, or lower prediction than the discriminative query for the second pair. The number within each cell indicates the number of relevant pairs. The lower half of the grid is symmetric to the upper half and therefore omitted.}}
\label{fig:comparisons_disc_generation}
\end{figure}

{To evaluate this discriminative-query implementation, we measure how accurately it predicts a participant's rating of a randomly-chosen example statement {(holding out this statement in the prompt, but including the other four example statements)}.
\Cref{fig:boxplot_disc_generation} displays the result of this analysis, for all 100 participants and all {5 choices of held-out statement (hence a total of 500 data points)}.
We observe that the means of the predictions are roughly calibrated, and that the 25th, 50th, and 75th percentile of the prediction distribution increase with the ground-truth rating.
The mean absolute error of the predicted rating is 0.93, less than one point on the 7-point scale.}

{
Since our democratic process is unaffected by monotone transformations of the rating scale, what matters is that the order of the predicted ratings is close to the order of the true ratings. To assess this, we consider pairs of predictions (each corresponding to a participant $i$ and a held-out statement $\alpha$) and check if $\disquery{i_1}{\alpha_1}$, $\disquery{i_2}{\alpha_2}$ have the same order as $u_{i_1}(\alpha_1)$, $u_{i_2}(\alpha_2)$.
Across all pairs with different ground-truth ratings, our discriminative query correctly orders 86.7\% (counting ties as one-half), far exceeding the 50\% accuracy of a constant predictor.
\Cref{fig:comparisons_disc_generation} breaks down this analysis based on the pair of ground-truth ratings.
We see that errors become less prevalent as we move away from the diagonal, which means that the discriminative query rarely produces the wrong ordering when the ground-truth ratings differ by more than one rating step.

One might wonder if this accuracy is explained by our query merely distinguishing popular from unpopular statements (or participants inclined to give high ratings from those who are not), but this is not the case:
when restricted to comparison pairs involving the same statement or the same participant, the accuracy remains about the same (83.3\% and 87.4\%, respectively), which shows that our discriminative query indeed captures the interaction between participant and rated statement.
}

\subsection{Generative Queries}
{We now turn to our implementation of the more ambitious generative query. Instead of passing the participants' free-form responses verbatim to the LLM, we first summarize each participant's free-form responses using an LLM.\footnote{{This summary consists of three bullet point lists: aspects of abortion most important to the participant, specific details and examples they mention, and details about the user's background. In addition, we include a global summary of the participants' opinion in two to three sentences. The summarization prompt can be found in \cref{app:summarization}.}}
In a previous version of this work, we adopted this summarization step to fit all 100 participants into the 32K token context window of an older version of GPT-4, allowing us to run \cref{alg:cjr} rather than the sampling-based \cref{alg:cjrsampling}.
Even with GPT-4o's expanded context window of 128K tokens, however, we retain this summarization since it might make useful information in participants' responses more accessible to the LLM and might reduce opportunities for bias based on participants' writing styles.
We do not include the responses to rating questions in the summaries, to avoid biasing the generation towards the example statements.
}

{We initially implemented the generative query with a single LLM call (given the summaries), but encountered challenges that persuaded us to adopt a multi-component design instead.
Most importantly, we did not succeed in prompting the LLM to effectively maximize the objective of the query $\genquery{S}{r}$, i.e., the $r$-highest utility among $S$.
For example, the statement generated by the LLM did not seem sufficiently responsive to the rank $r$.

Conceptually, the query $\genquery{S}{r}$ can be decomposed into two subtasks: \textbf{(i)}~identifying a subset $S' \subseteq S$ of $r$ agents amenable to agreeing on a statement, and \textbf{(ii)}~generating a statement that maximizes the minimum utility in $S'$.
When we asked the LLM to answer generative queries by performing these two subtasks in a chain of thought, the LLM seemed to struggle with subtask~(i) but to perform well on subtask~(ii).

Due to the LLM's difficulties in performing subtask~(i), our implementation of this subtask combines an LLM-generated feature embedding of agents with clustering, a more robust way of identifying aligned agents.
To generate the features, we randomly sample 50 of the bullet points occurring in the LLM-generated participant summaries, and then ask the LLM to rate the degree to which each agent is aligned with each feature on a 7-point scale.\footnote{{As in the discriminative query, we use the token probabilities to get an expected rating that need not be an integer.}}
These ratings embed each agent in a 50-dimensional Euclidean space, in which we use nearest-neighbor and balanced k-means clustering to identify cohesive sets of agents of a specified size.

Not only does subtask~(ii) more closely resemble a straightforward text generation task, which should play to the LLM's strengths, but there is also precedent in the literature for solving similar problems with LLMs: \citet{BCST+22} fine-tune an LLM to generate \emph{consensus opinions}, i.e., statements that are good compromises for a small group of agents.
Our setting is more challenging along several dimensions: our generation is based on much more detailed information about participants' opinions,\footnote{{In our experiment, participants wrote between 124 and 1324 words
(57––868 words excluding the the justification of ratings) within a median time of 26 minutes, whereas participants in \citet[App.\ B]{BCST+22} wrote between 10 and 200 words within a 5-minute limit on writing time. 
}}
and we aggregate over groups of up to twenty participants rather than up to five.
In addition, we aim to solve this task by prompting a general-purpose LLM rather than fine-tuning, albeit an LLM that is more advanced than the Chinchilla LLM used by \citet{BCST+22}.

We implement subtask~(ii) with a single chain-of-thought prompt to the LLM.
We instruct the LLM to write a statement with the goal of maximizing the minimum agreement among the given agents, to include points of agreement or likely agreement, and to avoid aspects that any of the agents fundamentally disagrees with.
Our chain-of-thought prompt first asks the LLM to identify common themes of participants' responses and key disagreements, and to judge for each aspect of the topic whether and how it should be included in the generated statement.
Only then should the LLM generate the statement, formulated as an opinion in the first person. See \cref{app:generative_query} for more details.

We implement the generative query as an ensemble, where we select several sets of agents through subtask~(i), apply the prompt from subtask~(ii) to turn each set into a statement, evaluate the objective value (see \cref{eq:generativedef}) of the resulting statements using the discriminative query, and then select the highest-scoring statement.
For our pilot experiment, our ensemble contains the following statement sources in each iteration of \cref{alg:cjr}:\footnote{{We also implemented a procedure to detect statements that did not express a subjective opinion, and filter them out from consideration in the ensemble. We took this measure after observing some LLM generations in our pre-tests that did not take a personal stance (e.g., ``Abortion is a sensitive and controversial topic that evokes strong emotions and differing opinions...''). Since none of the generated statements in our pilot experiment were flagged, and since we did not filter in the experiments described in this section, we omit a more detailed description.}}
\begin{itemize}
    \item We run balanced k-means clustering \citep{BBD00} to partition the remaining agents into clusters of size $n/k$, which yields $k$ sets of agents in the first round, $k-1$ in the second round, \dots, and one set in the last round.
    \item We also add a few smaller groups of agents to the ensemble, which are determined by sampling a random remaining agent and returning their nearest neighbors in the feature space. Specifically, we generate two sets of five agents and two sets of ten agents in this way. Given that these groups have fewer than $n/k$ members, we expect them to be more cohesive in their opinions, so subtask~(ii) is likely to generate more specific statements.
    \item We also allow our ensemble to select unused statements from previous calls to the generative query. This might surface good statements at no additional computational cost, since we have already invoked all the necessary discriminative queries.
\end{itemize}
}

{
Evaluating the generative queries quantitatively is challenging for several reasons.
First, since the 100 participants never see the statements generated from their responses, we do not know their real ratings for these statements and have to rely on the discriminative queries as a proxy.
Second, since the optimization over all possible statements in the idealized generative query (see \cref{eq:generativedef}) cannot be computed in practice, we lack a ground truth for how far our implementation is from the ideal query.
We can, however, compare our generated statements to a natural baseline, the statements written by the participants we are summarizing.

In light of the similarity of subtask~(ii) with the work of \citet{BCST+22}, we use an analogous evaluation scheme for our generative query implementation:
We randomly sample 1000 groups of four from the 100 participants and then apply our LLM prompt several times to generate candidate consensus statements.
(Since we query GPT-4o with a temperature of 1 for our generative prompts, the generated statement is random.)
We then compare these generated statements to the four statements written by each of the sampled participants.\footnote{{As mentioned above, these comparisons are made using our discriminative query, rather than real ratings.}}
Participants prefer a randomly-selected generated statement over the statement written by a random peer with 74.5\% probability (95\% confidence interval (CI): 73.8\% -- 75.2\%),\footnote{{In the following, all confidence intervals have 95\% confidence. Intervals for simple binomial proportions are computed with Jeffrey's intervals; other confidence intervals are computed via bootstrapping.}} and their mean score across three generated statements exceeds the mean score for the three peer statements with 87.2\% probability (CI: 86.0\% -- 88.4\%).
Though a comparison to the numbers by \citet{BCST+22} is not on equal terms,\footnote{{Most importantly, their comparisons are based on ground-truth preferences by the participant rather than a discriminative query/reward model. This makes their comparison harder in some ways (potentially higher variance in human responses, errors in the reward model). On the other hand, our statements come directly out of the LLM, whereas each of their generations is the best out of 16 generations as judged by reward models for the participants, which should help them. \citet{BCST+22} evaluate over a variety of policy questions with less detailed preferences, whereas we focus on one issue in detail. Finally, they discard a third of participants based on intra-rater reliability, whereas we only filter out a few participants with LLM-written or extremely low-effort responses.
}} this number (87.2\%) compares favorably to their reported 78\%.

Most relevant to the quality of subtask~(ii) of the generative query is the \emph{minimum} rating across the four participants.
As \cref{fig:win_rate} shows, a randomly-selected generated statement beats \emph{all four} participant statements with respect to the minimum utility with 61.1\% probability (CI: 58.0\% -- 64.1\%), and the best of four generated statements beats all four participant statements with 79.0\% probability (CI: 76.4\% -- 81.4\%).
This shows that, as judged by the discriminative query, our prompt for subtask~(ii) finds good consensus statements, clearly exceeding the baseline of taking the most broadly acceptable participant position, especially when we ensemble over a few generations.
}

\begin{figure}[t]
\centering
\includegraphics[width=.5\textwidth]{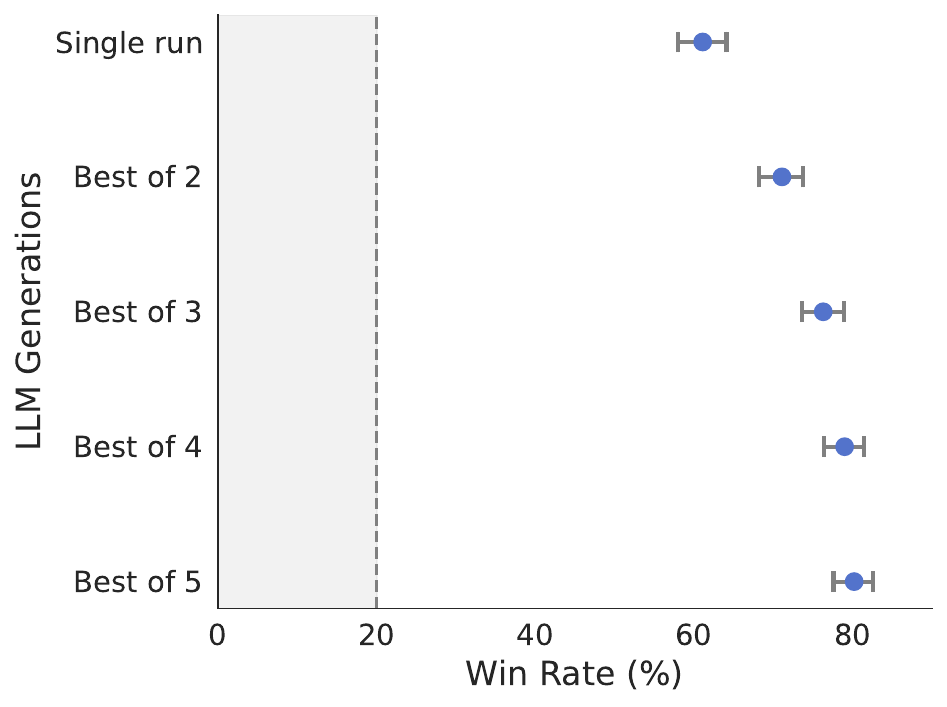}
\caption{{Win rate of LLM-generated statements over all four participant-written statements, in terms of the minimum predicted utility across the four participants. The dashed line indicates a win rate of 20\%, which would be obtained if the generated statement followed the same distribution as the human-written statements. Win rates are computed for 1000 random groups of four participants, and ranges indicate 95\% confidence intervals.}}
\label{fig:win_rate}
\end{figure}

{For subtask~(i)
we verify that, by ensembling over a few clusters, we consistently produce statements with high agreement rating.
For this experiment, we randomly draw 80 out of the 100 agents and attempt to find a statement that maximizes the 20th-highest rating.
This scenario simulates what is required of the generative query in the second round of running \cref{alg:cjr} with $n=100, k=5$. 
Recall that one source of statements in our ensemble samples a uniformly random agent, takes their 9 nearest neighbors, and then generates a statement for these 10 agents. In this experiment, we run this process four times, so we get four clusters and hence four statements in our ensemble. 
We run this experiment for 50 scenarios and show the results in \cref{fig:gen_query_eval}.

For about half of the sampled scenarios, even a single generation produces a statement whose top-20 rating is essentially maxed out, but there is still a decent chance of a poor generation (possibly due to the cluster's opinions not being well aligned with the preferences of participants outside of the cluster).
By ensembling over several clusters, we eliminate these bad generations, and almost always arrive at a statement for which at least 20 participants have close to the maximal rating.
This experiment shows that our implementation of the generative query, composed of both subtasks, produces high-quality statements.

\begin{figure}[t]
    \centering
    \includegraphics[width=0.5\textwidth]{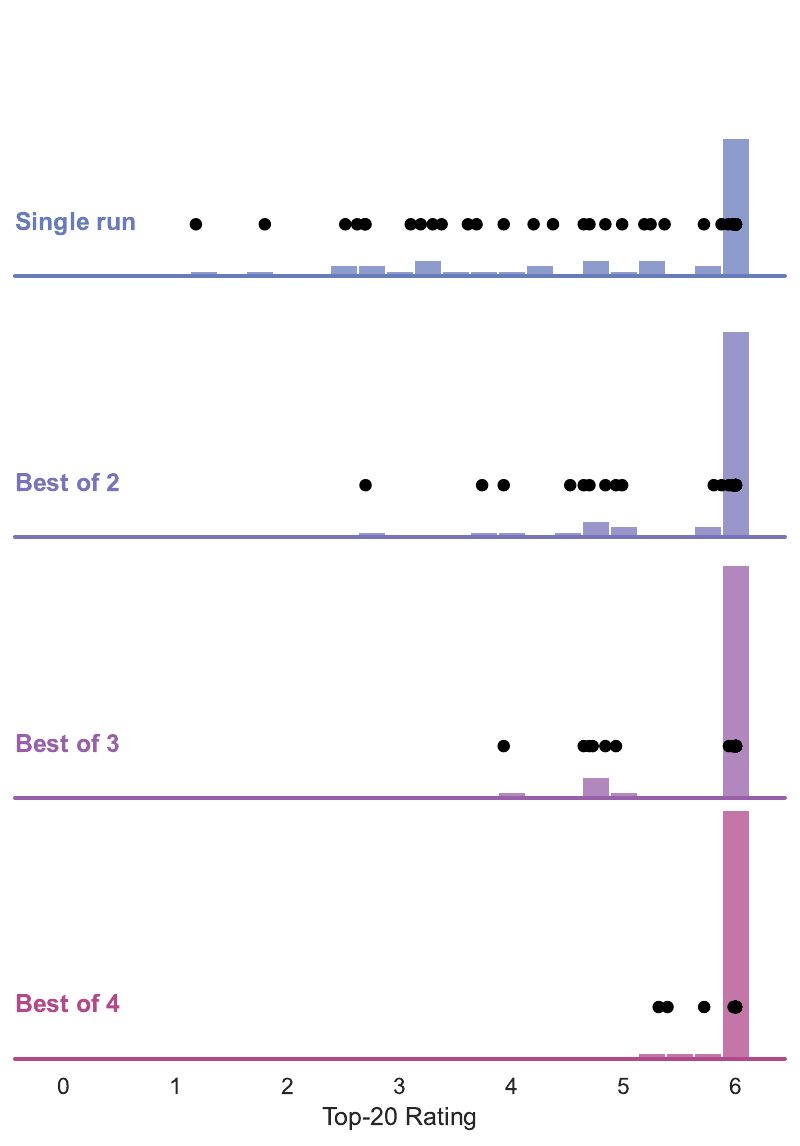}
    \caption{{Distribution of the top-20 rating of LLM-generated statements on 50 randomly subsampled sets of 80 agents, across different numbers of generation attempts.
    Statements were generated by applying the generation prompt to one, two, three, or four nearest-neighbor clusters of size 10, and taking the statement with highest top-20 rating.}
    }
    \label{fig:gen_query_eval}
\end{figure}

}

\section{A Study of Public Opinion on Abortion}
\label{sec:pilot}
{We piloted our democratic process to summarize public opinion on abortion.}
We ran surveys studying this topic on {August 8--9, 2024 and generated a representative slate of five statements.}

\subsection{Pilot Description}
We first recruit {a sample of} 100 participants through the online platform Prolific\footnote{\url{https://www.prolific.com/}}, which we refer to as the \emph{generation sample}.
{Our sample consists of US residents aged 22 and higher, stratified to reflect US residents in terms of gender and voting behavior in the 2020 presidential election.}\footnote{For more details on the demographic composition of the sample, see \cref{app:demographics}.}
{We ask these participants to complete a survey on abortion.}\footnote{See \cref{app:survey} for the verbatim survey questions.}
{Participants start reflecting on the topic by answering four questions in free-form text: how often they think about or discuss abortion, whether abortion should be legal or illegal, what their opinions on abortion are rooted in, and whether there are situations where they are uncertain about whether abortion is appropriate.
Then, we ask participants to summarize their stance on how society should deal with abortion, in a self-contained three-sentence statement.
}

We also ask participants to rate their agreement with {five} example statements, which we generated with a single call to {GPT-4o} and without knowledge of participant responses, {(see \cref{app:zero-shot} for the prompt).}
These ratings are given on the {seven}-level scale described at the beginning of \cref{sec:experiments}, and rating questions are shown to participants in random order.

{Based on participant responses, we extract a slate of five representative statements using \cref{alg:cjr} and our implementation of the queries.
As a baseline, we also ask GPT-4o to directly generate a slate of five statements, providing the LLM with the free-text responses of all 100 agents in its context window (see \cref{app:one-shot} for details).}

To evaluate {our} slate, we launch a second survey with a new set of 100 stratified participants, the \emph{validation sample}, to evaluate the slate's statements.
In this validation survey, we ask {participants to rate the ten statements from both our and the baseline slate} (in random order, using the same question format as for the ratings in the generation survey).
For reproducibility, and to support future research on online participation, we made participants' full responses publicly available at \url{https://github.com/generative-social-choice/survey_data/}.

\subsection{Results}
\label{sec:pilotresults}
{We defer the slate of five statements we generated to \cref{app:ourslate}, along with details about which components of the ensemble generated each statement.}\footnote{{When running our process on this data, calls to GPT-4o occasionally did not return correctly-formatted output during subtask (ii) of the generative query. These generations were simply ignored, occasionally leading to slightly smaller ensembles. Our ensemble approach is robust to these failures, and fixing them (by, e.g., retrying until a correctly-formatted output is returned) would only improve our results.}}
{Three of the statements express a clear pro-choice position and one statement a clear pro-life position (possibly with exceptions).
The remaining statement expresses discomfort with abortion, but advocates for the legality in cases of special hardship and for pursuing {measures outside the legal arena} to reduce the frequency of abortions.
This split of the slate is broadly in line with polls of US residents, which find that ``about six-in-ten (63\%) say abortion
should be legal in all or most cases'' \citep{Pew24}.
In \cref{app:ourslate}, we also check that the slate is plausibly proportional to the generation sample, by relating the statements on our slate to the example statements in the generation survey.
}

In the remainder of this section, we demonstrate the representativeness of the slate using the validation sample, a fresh sample of 100 US residents.
According to the ideal of proportional representation, each statement in our generated slate should represent 20\% of the US population as accurately as possible.
Following this principle, we match the participants of our validation sample to the statements of our slate such that each statement has 20 participants matched to it and {the mean participant rating} for their assignment is maximized, i.e., such that the balanced assignment maximizes the representation objective of \citet{Mon95}.
We then study the ratings of participants for their assigned statements.

\begin{figure}[h!]
    \centering
    \begin{minipage}[b]{0.47\textwidth}
        \centering
        \includegraphics[width=\textwidth]{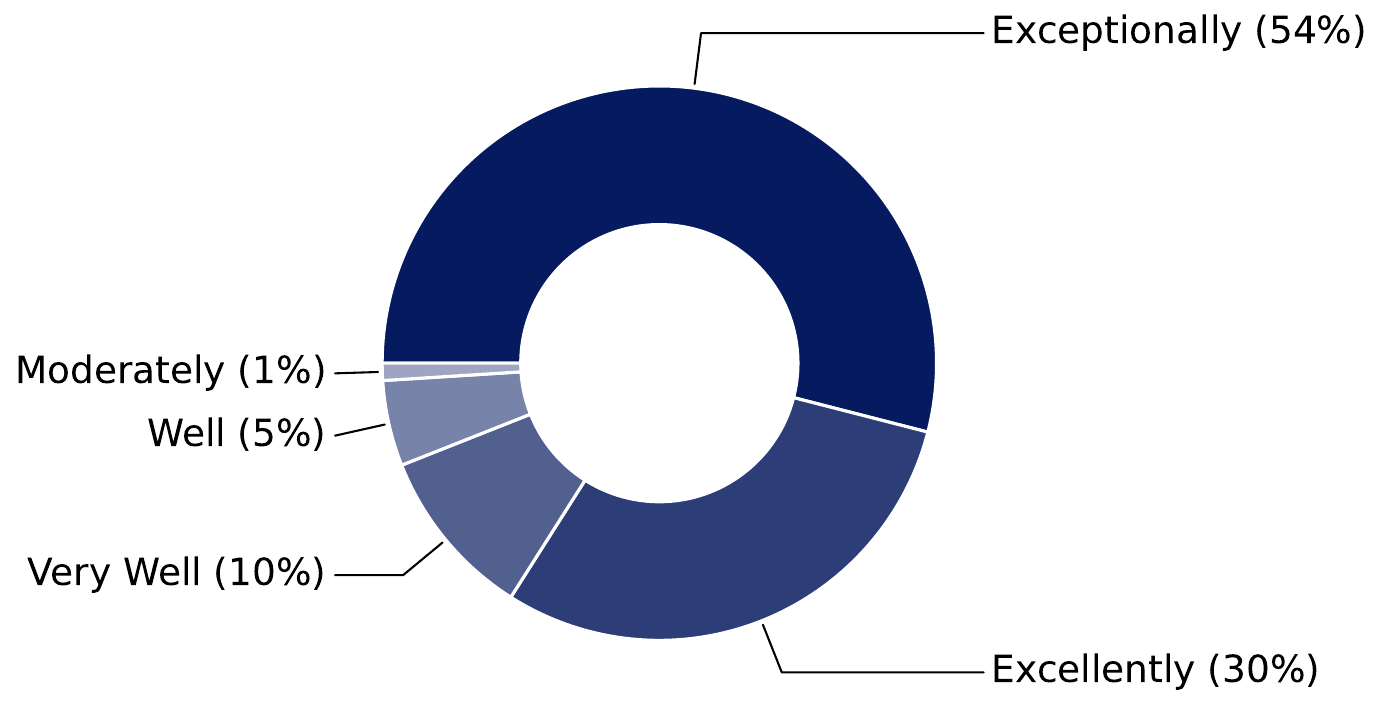}
        \caption*{Our Slate}
    \end{minipage}
    \hfill
    \begin{minipage}[b]{0.47\textwidth}
        \centering
        \includegraphics[width=\textwidth]{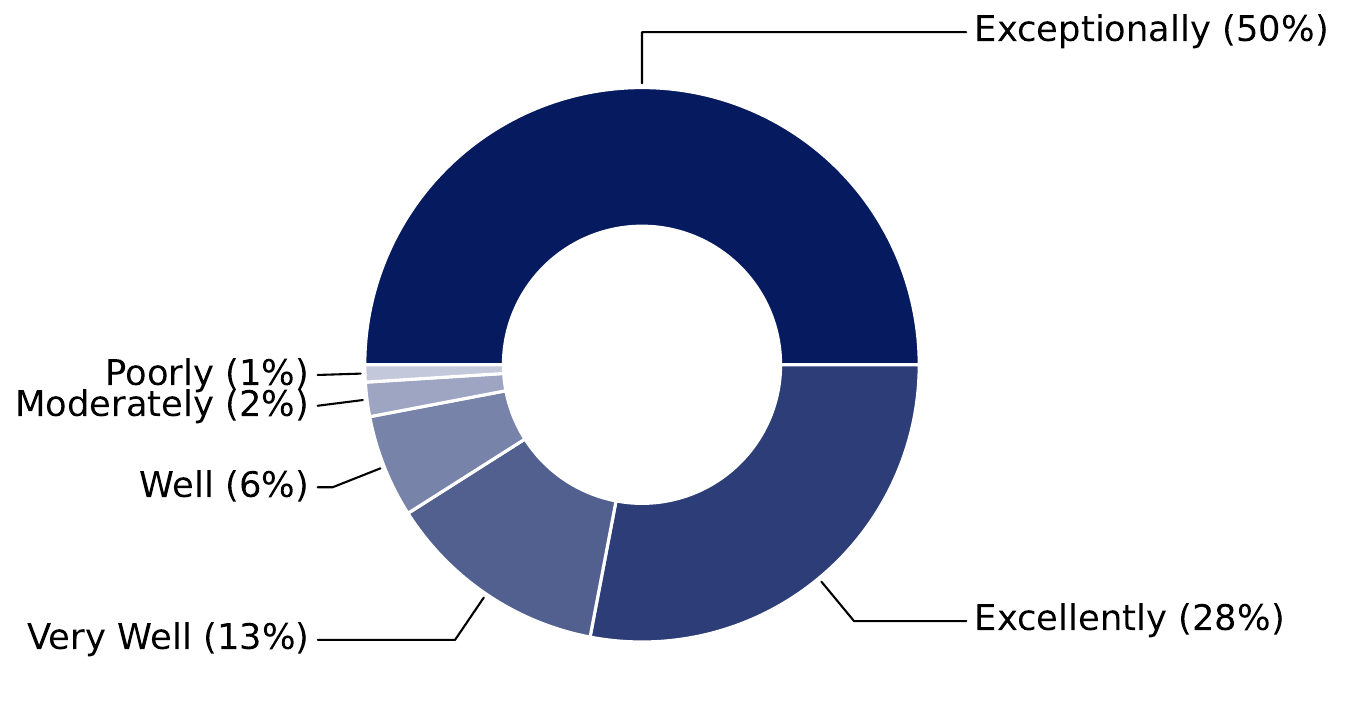}
        \caption*{Baseline Slate}
    \end{minipage}
    
    \vspace{0.5cm}
    \begin{minipage}{\textwidth}
        \centering
        \includegraphics[width=0.75\textwidth]{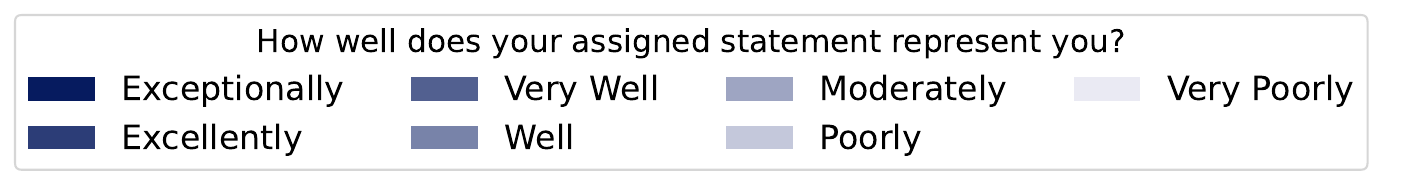}
    \end{minipage}

    \caption{Ratings of participants from the validation survey for their \emph{assigned} statement.}
    \label{fig:pie-chart}
\end{figure}

{We find that the mean rating of a participant for their assigned statement is 5.31, between the two highest level of our rating scale, ``excellently''~(5) and ``exceptionally''~(6).
As can be seen on the left-hand side in \cref{fig:pie-chart}, 
54 participants say that their assigned statement ``exceptionally'' captures their viewpoint, and an additional 30 participants say that it captures their viewpoint ``excellently''.
These levels far exceed the mean rating level of 2.74 in the survey (median: ``moderately'' (2), see \cref{app:ratingdistribution} for the distribution of ratings across all ten statements).
Only a single participant feels less than ``well''~(3) represented\emdash{}the midpoint of our scale\emdash{}by their assigned statement.}
Hence, the vast majority of participant opinions are represented accurately by our slate.

{
The baseline slate also performs well (see \cref{fig:pie-chart}, right), but obtains worse ratings than our slate:
the mean rating is 5.15, 50 participants are captured ``exceptionally,'' 28 participants ``excellently'' by their assigned statements, and 3 participants are represented less than ``well.''\footnote{{In fact, the ratings for our slate dominate those for the baseline slate in the following sense: for any rating level $\theta$, at least as many participants rate their assigned statement in \emph{our slate} at $\theta$ or better as in the baseline slate.}}
}

{
A key question is whether these slates satisfy our proportionality axiom BJR.
Though we cannot conclusively answer this question, we can approach it by investigating, for each slate, how close the statements of the other slate are to being BJR violations.
Specifically, we maximize the size of a coalition $S$ of agents such that these agents all rate some statement $\alpha$ from the other slate at some level $\theta$ or better, and all rate their assigned statement below $\theta$.
Note that a BJR violation is exactly such a coalition, if its size is at least $n/k = 20$.
For our own slate, the largest achievable coalition consists of five participants, for a statement $\alpha$ on the baseline slate (statement B5 in \cref{app:baselineslate}) that advocates for criminalization in all cases, and $\theta = 6$ (``exceptionally'').
For the baseline slate, coalitions are larger: two pro-choice statements on our slate (S2 or S4 in \cref{app:ourslate}) induce coalitions of size 9 (for $\theta = 6$).
This is another indication that our generated slate is more representative of the validation sample, and more likely to satisfy BJR.\footnote{{A coarser measure of representativeness, but one that might be easier to interpret, is to simply ask how many participants strictly prefer some statement $\alpha$ that is not on the slate over their assigned statement. The maximum such number for our slate and a baseline statement $\alpha$ is 6; for the baseline slate and a statement $\alpha$ from our slate, it is 12.}}
}

\section{Discussion}
\label{sec:discussion}

As a result of the increase in power, availability, and steerability of LLMs, we are currently witnessing an explosion of creative prototypes for {participatory} processes with generative-AI components~\cite[e.g.,][]{MM23,KSI+23,SPJ+23,DKM+23}.
This expansion of the capabilities of participation is thrilling, but\emdash{}as these prototypes continue to proliferate and eventually turn into deployed {applications}\emdash{}we ought to critically interrogate {their legitimacy} on two fronts.

The first line of questioning has already received broad attention~\cite[e.g.,][]{SVDH+23}: \emph{can the AI building blocks in the process be trusted?}
Taking our process as an example, we have started answering this question by measuring the average accuracy of our LLM queries, by overcoming an observed lack of robustness through the ensemble implementation of our generative query, and by piloting the process in practice.
Before our process is ready for high-stake deployments, though,
it must yet be hardened against malicious participant input (e.g., prompt injections~\citep{Wallace2019-vo} meant to unduly sway generative queries), and the effect of biases against groups of people~\citep{Basta2019-si,Kurita2019-zx} and viewpoints~\citep[e.g.][]{HSW23} in the LLM must be studied and counteracted.
{Continued research in these areas will likely lead to ever more powerful  and robust LLMs, approximating our oracle assumption more and more closely. Nevertheless, we propose to ``trust, but verify'' in high-stakes settings: Once the democratic process has produced a slate of statements, one could let participants vote on the statements proposed by other participants, alongside the slate produced by the our process. The slate should then only be adopted if no participant's statement witnesses a BJR violation with respect to the slate.\footnote{{If there is a violation, one might repeat the process, or select a slate that satisfies BJR with respect to the comments explicitly voted on. In this way, the LLM-augmented process does no worse than a classical process, even for arbitrary corruptions of LLM outputs.}} In this way, the democratic process can tap into the power of LLMs while ensuring that the voters, not machines, have the final word.}

{Even if the AI building blocks are trustworthy, another question remains:} \emph{is the process around the AI components democratic?}
Granted, AI participation processes typically solicit input from all participants, and might even treat participants symmetrically, but that property alone (\emph{neutrality}) is clearly not sufficient.
We believe that voting rules with AI elements, just like those without, should argue their case based on axioms that ensure, for example, the rule's responsiveness, efficiency, and fairness.

At its heart, generative social choice articulates a vision of what it means for an AI-enhanced voting rule to be democratic.
By showing the required ingredients\,---\,the axioms targeted by the rule, necessary conditions on the behavior of the LLM, and evidence that the LLM meets these conditions\,---\,a voting rule can assuage the above two threats to legitimacy, while tapping into the possibilities enabled by generative AI.

\section*{Acknowledgments}
We thank Nika Haghtalab and Abhishek Shetty for pointers on how to apply sampling bounds to sampling without replacement. 
This work was partially supported by OpenAI through the ``Democratic Inputs to AI'' program and by the Office of Naval Research under grant N00014-20-1-2488.
Manuel W\"uthrich was partially funded by the Swiss National Science Foundation (SNSF). 
Paul Gölz was supported by the National Science Foundation under Grant No. DMS-1928930 and by the Alfred P. Sloan Foundation under grant G-2021-16778 while in residence at the Simons Laufer Mathematical Sciences Institute (formerly MSRI) in Berkeley, California, during the Fall 2023 semester. Sara Fish was supported by an NSF Graduate Research Fellowship and a Kempner Institute Graduate Fellowship.

\bibliographystyle{ACM-Reference-Format}
\bibliography{abb,ultimate,more}

\appendix

\section{Relationship Between BJR and Other Justified Representation Axioms}
\label{app:incomparablebjr}
Our notion of BJR is closely related to several axioms in the social choice literature.\footnote{Note that we defined slates as multisets, whereas these axioms typically define committees as sets. The discussion in this section is both valid if one translates the multi-winner axioms into the multiset setting, or by using the set variant of BJR described in Footnote~6 in the body.}
Suppose for the time being that we were to relax BJR by not requiring the matching of agents to statements to be balanced, in which case each agent would be matched to their most preferred statement without loss of generality.
In the subsetting of approval utilities, this relaxed axiom coincides with the \emph{justified representation (JR)} axiom of~\citet{ABCE+17}.

For our setting of general cardinal utilities, the relaxed axiom is implied by \emph{extended justified representation (EJR)} and \emph{full justified representation (FJR)} as defined by \citet{PPS21}.

\begin{table}[htb]
    \vspace{\topskip}
    \begin{minipage}[t]{.45\linewidth}
      \vspace{-\topskip}
      \centering
        \caption{Utility matrix of first example instance, with $k=n=3$.}
        \label{tbl:jrbad}
        \begin{tabular}{lllll}
            \toprule
            & $\alpha$ & $\alpha'$ & $\beta$ & $\beta'$ \\
            \midrule
            $u_1$ & 1 & 1 & 0 & 0 \\
            $u_2$ & 1 & 1 & 0 & 0 \\
            $u_3$ & 0 & 0 & 1 & 1 \\
            \bottomrule
        \end{tabular}
    \end{minipage}\hfill
    \begin{minipage}[t]{.45\linewidth}
        \vspace{-\topskip}
          \centering
      \caption{Utility matrix of second example instance, with $k=n=2$.}
      \label{tbl:fjrbad}
\begin{tabular}{lllll}
\toprule
& $\alpha_1$ & $\alpha_2$ & $\beta$ & $\beta'$ \\
\midrule
$u_1$ & 3 & 0 & 2 & 2 \\
$u_2$ & 0 & 3 & 2 & 2 \\
\bottomrule
\end{tabular}
    \end{minipage}
\end{table}
The need for a new, balanced-matching-based notion of justified representation is best explained using two simple examples.
The first example, given in \cref{tbl:jrbad}, is standard: $k=3$ statements must be selected, two-thirds of agents (specifically, agents 1 and 2) approve statements $\alpha, \alpha'$, and the remaining third of agents (agent 3) approves statements $\beta, \beta'$.
As has been frequently observed \citep[e.g.,][Example 3]{ABCE+17}, JR (and thus the relaxation of BJR with unbalanced matchings) is satisfied by the slate $\{\alpha, \beta, \beta'\}$.
This is problematic since this slate is patently unproportional: it represents two-thirds of the population by one-third of the slate, and vice versa.

JR cannot rule out this form of unproportionality because each member of the two-thirds bloc is already represented by some statement they approve, and JR does not allow agents and coalitions to formulate any claims to representation beyond that point.
Axioms like EJR and FJR allow coalitions to make stronger claims than JR by assuming that 
an agent (say, agent 1 in the previous example) may prefer to be represented by \emph{multiple} statements rather than just one. Specifically, these axioms model an agent's utility as being the sum of their utilities for all statements on the slate.

Though this approach allows EJR and FJR to rule out the unproportional slates in the first example, it causes them to require slates on other instances that we find undesirable for the setting of statement selection, especially for non-approval utilities.
\Cref{tbl:fjrbad} shows one such instance, in which two statements must be selected for two agents.
Each agent $i \in \{1, 2\}$ has a statement $\alpha_i$ {that} is very specific to $i$ and thus has a high utility for $i$ but low utility for the other agent.
In this instance, we believe that a slate consisting of these two statements would be a good choice since it represents the specificity of agents' preferences to the highest degree; indeed, only this slate satisfies BJR.
EJR and FJR, by contrast, rule out these statements, since they prefer to represent both agents jointly by two less specific statements (namely, $\beta, \beta'$) rather than each agent individually by a specific statement.\footnote{One might hope that EJR and FJR can be adapted to this perspective, by extending utilities to sets in a unit-demand rather than additive way. With this modification, however, they no longer rule out the unproportional slate in the first example instance.}

Our axiom of BJR enforces more specificity on the second instance, while ruling out the unproportional slates on the first example instance.
Instead of allowing a single agent to be represented by multiple statements, BJR's analysis of the shortcoming of JR in the first example is that too many agents were represented by a single statement on the slate.
Philosophically, we see connections between our axiom and the notion of \emph{fully proportional representation} of \citet{Mon95}: ``voters should be segmented into equal-sized coalitions, each of which is assigned a representative, such that the preferences of voters are as closely as possible reflected by the representatives of their segment.''
In the remainder of this appendix, we show that BJR, other than implying JR, is incomparable to previously studied notions of justified representation, even in the setting of approval utilities.
{For a definition of these axioms, we refer the reader to Definitions 4.3, 4.5, 4.7, and 4.10 of \citet{LS23a}.}

\begin{proposition}
Balanced justified representation (BJR) is incomparable with proportional justified representation (PJR), extended justified representation (EJR), full justified representation (FJR), and core stability.
This incomparability holds even for approval utilities and holds both in our setting where slates/committees are multisets\footnote{\citet{BGP+22} give a formal embedding to translate existing justified representation axioms to the multiset setting (``party-approval elections'', in their terminology). Whereas the existence of core stable committees is unresolved when committees are sets of alternatives, such committees are guaranteed to exist in the multiset setting~\citep{BGP+22}.} and in the classical setting where they are sets (using the adaptation of BJR from Footnote~6 in the body).
\end{proposition}
\begin{proof} We will show this incomparability in two steps: we first show that BJR implies none of the other axioms, and then that none of the axioms implies BJR.
\paragraph{BJR does not imply other axioms.}
Consider the instance with $n=6$, $k=4$, and the following utilities:
\begin{center}
\begin{tabular}{ccccccc}
\toprule
& $\alpha$ & $\alpha'$ & $\alpha^-$ & $\beta$ & $\gamma$ & $\delta$ \\
\midrule
$u_1$ & 1 & 1 & 1 & 0 & 0 & 0 \\
$u_2$ & 1 & 1 & 1 & 0 & 0 & 0 \\
$u_3$ & 1 & 1 & 0 & 0 & 0 & 0 \\
$u_4$ & 0 & 0 & 0 & 1 & 0 & 0 \\
$u_5$ & 0 & 0 & 0 & 0 & 1 & 0 \\
$u_6$ & 0 & 0 & 0 & 0 & 0 & 1 \\
\bottomrule
\end{tabular}
\end{center}

In this instance, the slate $\{\alpha^-, \beta, \gamma, \delta\}$ satisfies BJR since, if we assign agents 1 and 2 to $\alpha^-$, agents 3 and 4 to $\beta$, agent 5 to $\gamma$, and agent 6 to $\delta$, then only agent 3 is not already maximally satisfied.
As a result, no potential deviating coalition can include the necessary $n/k = 3/2$ agents.

By contrast, this slate does not satisfy PJR because the coalition of agents 1, 2, and 3 is large enough to proportionally claim $\ell = 2$ statements, has two statements they all like in common ($\alpha, \alpha'$), but only one of the four statements on the slate is liked by any agent in this coalition. 

Since EJR, FJR, and core stability imply PJR, none of them can be implied by BJR either.

\paragraph{Other axioms do not imply BJR.}

To prove this direction of the claim, consider the following instance with $n=8$ agents and $k=4$.
The table below shows the agents' utilities for a subset of the statements:
\begin{center}
\begin{tabular}{ccccc}
\toprule
& $\alpha$ & $\alpha'$ & $\beta$ & $\beta'$ \\
\midrule
$u_1$ & 1 & 1 & 0 & 0 \\
$u_2, u_3, u_4$ & 1 & 0 & 0 & 0 \\
$u_5$ & 0 & 0 & 1 & 1 \\
$u_6, u_7, u_8$ & 0 & 0 & 1 & 0 \\
\bottomrule
\end{tabular}
\end{center}
In addition, any pair of agents $\{i, j\}$ is associated with a statement $\gamma_{i, j}$, which exactly they approve.

In this instance, the slate $\{\alpha, \alpha', \beta, \beta'\}$ does not satisfy BJR.
Indeed, since a balanced assignment assigns two agents to each statement of the slate, it holds for any such balanced assignment that some agent $i$ assigned to $\alpha'$ and some agent $j$ assigned to $\beta'$ have 0 utility for their assigned statement.
Since these two agents could deviate to the statement $\gamma_{i,j}$, BJR is violated.

By contrast, we will show that this slate satisfies core stability, and thus the weaker axioms of FJR, EJR, and PJR.
Indeed, suppose that some non-empty coalition $S$ along with a (multi)set $T$ of at most $\frac{|S|}{n} \cdot k$ statements formed a core deviation.
Suppose that $S$ includes $0 \leq x \leq 2$ many among the agents $\{1, 5\}$.
Since agents 1 and 5 have a utility of 2 for the candidate slate, they can only be part of a deviating coalition if the deviation $T$ gives them utility at least 3.
Analogously, since the other agents have a utilty of 1 for the candidate slate, they can only deviate if $T$ gives them utility at least 2.
If we define the \emph{coalition welfare} $\mathit{cw}$ as the sum of utilities, across the agents in $S$, for $T$, it follows that $\mathit{cw} \geq 3\,x + 2\,(|S| - x) = 2\,|S| + x$.
Now, the average contribution of a statement in $T$ to this objective is
\begin{equation}
\frac{\mathit{cw}}{|T|} \geq \frac{2 \, |S| + x}{|S| \, k / n} \geq \frac{2 \, n}{k} + x \, \frac{n}{|S| \, k} = 4 + x \, \underbrace{\frac{2}{|S|}}_{> 0}. \label{eq:cw}    
\end{equation}
Note that statements $\alpha$ and $\beta$ are the only ones that can potentially contribute at least 4 to the coalitional welfare (since all other statements are approved by fewer than two agents), and they can also contribute only exactly an amount of 4, never more.
Thus, it must be that $\mathit{cw}/|T|$ is equal to 4.
This, in turn, implies that $x=0$, i.e., that agents 1 and 5 are not in $S$, and that $T$ consists only of the statements $\alpha$ and $\beta$ (possibly with repetition).
But now observe that, since agents 1 and 5 are not in the coalition, $\alpha$ and $\beta$ cannot marginally contribute more than 3 to the coalition welfare, which contradicts \cref{eq:cw} and thus shows that the slate satisfies core stability.
\end{proof}

{
Finally, one might hope to strengthen BJR by allowing coalitions to deviate as long as any member of the coalition strictly increases their utility in the deviation, rather than requiring agents to cross a common threshold $\theta$:
\begin{definition}
    A slate $W$ is in the \emph{balanced unit-demand core} if there is a function $\omega: N\to W$, matching agents to statements such that each statement on the slate is matched to $\lfloor n/k \rfloor$ or $\lceil n/k\rceil$ agents, for which there is no coalition $S\subseteq N$ and statement $\alpha\in\mathcal{U}$ such that \emph{(i)}~$|S|\geq n/k$ and \emph{(ii)}~$u_{i}(\alpha) > u_{i}(\omega(i))$ for all $i\in S$.
\end{definition}

Unfortunately, the resulting axiom may not be satisfiable.
{The possible emptiness of the core can be seen using the instance~\cite[App.\ C]{FMS18} previously used to show the analogous statement for additive, non-approval preferences:}

\begin{proposition}
The balanced unit-demand core can be empty.
\end{proposition}
\begin{proof}
Consider the following instance with six agents, six statements, and $k=3$:
\begin{center}
\begin{tabular}{ccccccc}
    \toprule
    & $\alpha$ & $\beta$ & $\gamma$ & $\alpha'$ & $\beta'$ & $\gamma'$ \\
    \midrule
    $u_1$ & 2 & 0 & 1 & 0 & 0 & 0 \\
    $u_2$ & 1 & 2 & 0 & 0 & 0 & 0 \\
    $u_3$ & 0 & 1 & 2 & 0 & 0 & 0 \\
    $u_1'$ & 0 & 0 & 0 & 2 & 0 & 1 \\
    $u_2'$ & 0 & 0 & 0 & 1 & 2 & 0 \\
    $u_3'$ & 0 & 0 & 0 & 0 & 1 & 2 \\
    \bottomrule
\end{tabular}
\end{center}
Note that the instance decomposes into two, isomorphic instances (one with primed variables and one with non-primed variables).
By symmetry, we can assume w.l.o.g.\ that the slate contains at most one non-primed statement and that, if there is a non-primed statement in the slate, this statement is $\alpha$.
That is, we can assume that neither $\beta$ nor $\gamma$ are on the slate.
Consider the coalition $\{2, 3\}$, and note that agent~2 does not value any statement on the slate better than $1$ and that agent~3 does not value any statement on the slate better than $0$, so the statement they are assigned to has at most this value.
Hence, these $2 = n/k$ agents would like to deviate to the statement $\beta$, which has value $2 > 1$ for agent~2 and value $1 > 0$ for agent~3.
This demonstrates that no slate is in the balanced unit-demand core.
\end{proof}

Hence, BJR cannot be strengthened to account for deviations without a uniform threshold while {retaining} feasibility.
The uniform-threshold technique (for generalizing approval-based proportionality notions without losing feasibility) was previously used by \citet{PPS21} to generalize EJR and FJR to general additive utilities.
}

\section{Deferred Proofs}\label{app:proofs}

\propcjr*

\begin{proof}
In this proof, we will use $\alpha_j, T_j$ to denote the values of $\alpha$ and $T$ assigned in a given iteration $1 \leq j \leq k$.
We construct the matching $\omega$ by, for each round $j = 1, \dots,k$, mapping all agents that were removed from $S$ in that round to the statement that was added to $W$ in that round, i.e. for all $i\in T_j$ we have $\omega(i)=\alpha_j$.
Clearly, this matching is balanced, since either $\lfloor n/k \rfloor$ or $\lceil n/k \rceil$ agents are removed in each round.

Now consider a coalition $S'\subseteq N$, a statement $\alpha'\in\mathcal{U}$, and a threshold $\theta \in \mathbb{R}$ such that $|S'| \geq n/k$ (and, by integrality, $|S'| \geq \lceil n/k \rceil$) and $u_{i}(\alpha')\geq \theta$ for all $i\in S'$.
Once Process~1 terminates we have $S=\emptyset$, hence there must be an earliest iteration $j$ where some agent $i' \in S'$ appeared in $T_j$.
At the beginning of iteration $j$ of the loop, it must thus still hold that $S' \subseteq S$.
Note that
\begin{align*}
\rthlargest{\lceil \bar{r} \rceil}(\{ u_i(\alpha') \mid i \in S\}) &= \rthlargest{\lceil n/k \rceil}(\{ u_i(\alpha') \mid i \in S\}) \geq \rthlargest{|S'|}(\{ u_i(\alpha') \mid i \in S\})\\
&\geq \rthlargest{|S'|}(\{ u_i(\alpha') \mid i \in S'\}) \geq \theta.
\end{align*}
Thus, since $i' \in T_j$ and by the definition of the generative query, it must hold that
\begin{align*}
u_{i'}(\omega(i'))=u_{i'}(\alpha_j) \geq \rthlargest{\lceil \bar{r} \rceil}(\{ u_i(\alpha_j) \mid i \in S\}) \geq \theta.
\end{align*}
We conclude that $S', \alpha', \theta$ do not violate BJR.
\end{proof}

\propjrlower*

\begin{proof}
Set $t \coloneqq n/k \, (1 - 1/k)$.
Let $n$ be some multiple of $k^2$, so that $t$ is an integer.
Suppose that there is one ``popular'' statement $\alpha$, which has utility 1 for all agents.
Furthermore, for each set $S$ of at most $t$ agents, let there be an ``unpopular'' statement with utility 1 for $S$ and 0 for all other agents.
This unpopular statement is a valid answer for any query of the shape $t$-$\genquery{S}{\cdot}$, because the $r$-th largest utility among $S$ for this statement is 1, the maximum possible utility of this instance.
Thus, with the right tie breaking, one can implement all $t$-$\genquery{\cdot}{\cdot}$ queries to return unpopular statements, from which it follows that the process will have to return a slate $W$ entirely of unpopular comments.

Since each unpopular statement has positive utility for at most $t$ agents, at most $k\cdot t = n \, (1-1/k) = n - n/k$ agents receive positive utility from any statement in $W$.
In other words, $n/k$ agents have utility 0 for all statements in $W$, but have utility 1 for the popular statement $\alpha$. This demonstrates a violation of (balanced) justified representation.
\end{proof}

\thmjrlowerpolytime*

\begin{proof}
    Choose $k$ to be an even integer and $n$ as a multiple of 8, such that $t \coloneqq n/8$ is integer as well.
    Fix a process that makes fewer than $\frac{2}{k}\,e^{n/(12k)}$ many $t$-$\genquery{\cdot}{\cdot}$ {queries} and any number of discriminative queries. We will prove the claim using the probabilistic method: we will define a random instance and show that the process will fail BJR with positive probability, which means that there exists a deterministic instance where the process fails BJR.
    In fact, the random instances we construct will have approval utilities, and we will derive a contradiction to not just BJR, but also JR on this instance, to simultaneously prove the ``this holds even\dots'' part of the claim.

    For given $n,k$, construct our instance as follows: Each set $S$ of $\frac{n}{2\,k}$ many agents has infinitely many ``unpopular'' statements that have utility 1 for $S$ and utility 0 for all other agents.
    Furthermore, each agent is uniformly and independently assigned a color in $\{1,2,\dots,k/2\}$, and all agents with the same color $c$ have utility 1 for a ``popular'' statement $\beta_{c}$, which has utility 0 for everyone else.
    Since all utilities are 0 or 1, there will typically be many statements $\alpha$ that are tied in the definition of a generative query $\genquery{S}{r}$: if there exist statements that have utility 1 for at least $r$ agents in $S$, any such statement may be returned; if no such statements exist, the query may return any arbitarary statement.
    To resolve this ambiguity, we assume that the generative query breaks ties in the “most favorable” way: the generative query will respond to $\genquery{S}{r}$ with a statement that has utility 1 for as many agents in $S$ as possible, and breaks remaining ties according to some canonical ordering of statements in which unpopular comments precede popular comments.

    Consider the trajectory of the process on an instance with just the unpopular statements, i.e., where each $t$-$\genquery{S}{\cdot}$ query of the process is answered by a canonical unpopular statement that attains the maximum number $\min(|S|, \frac{n}{2 \, k})$ of agents in $S$ that have utility 1 for it.
    
    Now, consider the random instance with unpopular and popular statements.
    We will show that, with positive probability, all $t$-$\genquery{\cdot}{\cdot}$ queries made by the process are still answered by their canonical unpopular statement, which means that the process will follow the same trajectory as above.
    This will be the case if, for each $t$-$\genquery{S}{\cdot}$ query made by the process and for each color $c$, at most $\frac{n}{2\,k}$ agents in $S$ have color $c$, so that $\beta_c$ will not be returned by the query.
    For a specific $S$ and $c$, the probability of this event can be upper-bounded using Chernoff as
    \begin{align*}
               & \mathbb{P}\left[\text{at least \ensuremath{\frac{n}{2\,k}} agents in \ensuremath{S} have color \ensuremath{c}}\right] \\
        ={}    & \mathbb{P}\left[\mathsf{Binomial}(n/8, 2/k) \geq 2 \cdot\frac{n}{4 \, k}\right] \\
        \leq{} & \exp\left[- \frac{n}{12 \, k}\right].
    \end{align*}
    By a union bound, it follows that, with positive probability, this event does not occur in any of the fewer than $\frac{2}{k}\,e^{n/{(12k)}}$ queries, for any of the $\frac{k}{2}$ colors. This implies that there is an instance in the support of our random instance on which the trajectory of the process remains the same as if there were no popular statements and where, in particular, the process must return a slate of unpopular statements.

    Finally, we show that, when the process only returns unpopular statements, it must violate justified representation.
    (This always hold for our random instance, ex post.)
    Since each unpopular statements gives positive utility to at most $\frac{n}{2\,k}$ agents, no more than $\frac{n}{2}$ agents can be covered by the slate of $k$ statements selected by the process.
    Therefore, there are at least $\frac{n}{2}$ uncovered agents, which are partitioned in some arbitrary manner across the $\frac{k}{2}$ many colors. By an averaging argument, there must be some color $c$ with at least $\frac{n}{k}$ uncovered agents, which means that the process' output violates justified representation and BJR for $\beta_{c}$.
\end{proof}

\begin{lemma}[Agnostic PAC learning for sampling without replacement] \label{lem:pacwithoutreplacement} Let $\mathcal{H}$ be a hypothesis class, consisting of binary classifiers $h:\mathcal{X}\to\mathcal{Y}$, with $|\mathcal{Y}|=2$, over some domain $\mathcal{X}$. Let $d<\infty$ denote the VC dimension of $\mathcal{H}$. For a given hypothesis $h\in\mathcal{H}$, denote its 0--1 loss on a nonempty finite set $S\subseteq\mathcal{X}\times\mathcal{Y}$ of labeled datapoints by $L_{S}(h)\coloneqq\sum_{(x,y)\in S}\bone\{h(x)\neq y\}/|S|$.

Let $D\subseteq\mathcal{X}\times\mathcal{Y}$ be a finite set of labeled datapoints. Consider a random process that chooses some number $m\leq|D|/2$ of labeled datapoints $S=\{(x_{1},y_{1}),(x_{2},y_{2}),\dots,(x_{m},y_{m})\}$ from $D$ uniformly and \textbf{without replacement}, and denote by $\hat{h}$ the empirical risk minimizer $\argmin_{h\in\mathcal{H}}L_{S}(h)$. For any $0<\epsilon<1,0<\delta<1$, this process will satisfy 
\begin{equation}
L_{D}(\hat{h})\leq\min_{h\in\mathcal{H}}L_{D}(h)+\epsilon\label{eq:lossbound}
\end{equation}
and 
\begin{equation}
\left|L_{S}(h)-L_{D}(h)\right|\le\epsilon\quad\forall h\in\mathcal{H}\label{eq:uniformbound}
\end{equation}
with probability at least $1-\delta$, as long as 
\begin{equation}
m\geq C\cdot\frac{d+\log1/\delta}{\epsilon^{2}}\label{eq:mbound}
\end{equation}
for some absolute constant $C$. \end{lemma} 
\begin{proof}
If $|D|\geq m^{2}/\delta$, the result will follow from the sampling bounds for i.i.d.\ samples. Note that we can implement the without-replacement drawing of $S$ through rejection sampling, i.e., by drawing a sample of $m$ datapoints uniformly \emph{with} replacement, and re-drawing if this sample should contain any datapoint multiple times. We will consider only the first round of this rejection sampling. The probability that any two datapoints are identical is at most $\sum_{i=0}^{m-1}i/|D|=\frac{m\,(m-1)}{2\,|D|}\leq\frac{m^{2}}{2\,|D|}\leq\delta/2$, so we reject with probability at most $\delta/2$. Moreover, since drawing with replacement is the same as drawing i.i.d.\ from the uniform distribution over $D$, we can apply a standard agnostic PAC learning bound~\citep[Thm.\ 6.8]{Shalev-Shwartz2014-tv} to show that the empirical risk minimizer $\hat{h}$ on the sample with replacement satisfies \cref{eq:lossbound} with probability at least $1-\delta/2$ as long as the constant in \cref{eq:mbound} is sufficiently large. By a union bound over both events, with probability at least $1-\delta$, the with-replacement sample is not rejected and additionally satisfies \cref{eq:lossbound}, which proves the claim for our sampling process without replacement in the case of $|D|\geq m^{2}/\delta$.

From here on, suppose that $|D|<m^{2}/\delta$. Essentially, our claim will follow from Theorem~2 by \citet{EP09b}, a bound on transductive learning, but we have to do some work to get their bound into our desired shape. We apply their Theorem~2 twice, with a value of $\delta$ that is half of the $\delta$ in our theorem, the full sample $D$, the hypothesis class $\mathcal{H}$, $\gamma=1$, and setting $m$ once to $m$ and once to $|D|-m$ (swapping the role of sampled and not sampled datapoints). By union-bounding over both invocations and unfolding some definitions in the theorem, we obtain that, with probability at least $1-\delta$, it holds for all $h\in\mathcal{H}$ that 
\begin{equation}
L_{D\setminus S}(h)\leq L_{S}(h)+R_{\mathit{trans}}(\mathcal{H})+\mathit{slack}\quad\text{and}\quad L_{S}(h)\leq L_{D\setminus S}(h)+R_{\mathit{trans}}(\mathcal{H})+\mathit{slack}\label{eq:twosideloss}
\end{equation}
where $R_{\mathit{trans}}(\mathcal{H})$ denotes the \emph{transductive Rademacher complexity} of $\mathcal{H}$ on $D$, and $\mathit{slack}$ is defined and bounded in the following.

The slack term is defined as 
\[
\mathit{slack}\coloneqq c_{0}\,q\,\sqrt{m}+\sqrt{\frac{s\,q}{2}\,\ln1/\delta},
\]
where $c_{0}<5.05$ is an absolute constant, $q\coloneqq\frac{1}{m}+\frac{1}{|D|-m}\leq\frac{2}{m}$, and $s\coloneqq\frac{|D|}{(|D|-1/2)\cdot(1-\frac{1}{2(|D|-m)})}$. Since $m$ is a positive integer, $m\geq1$, hence $|D|-m\geq m\geq1$, and thus $s=\frac{|D|}{|D|-1/2}\cdot\frac{1}{1-\frac{1}{2\,(|D|-m)}}\leq4/3\cdot2=8/3$. Thus, 
\begin{equation}
\mathit{slack}\leq\frac{5.05\cdot2}{\sqrt{m}}+\sqrt{\frac{8/3}{m}\,\ln1/\delta}=\frac{1}{\sqrt{m}}(10.10+\sqrt{8/3\,\ln1/\delta}).\label{eq:slackbound}
\end{equation}

Next, we bound the transductive Rademacher complexity, for which we require several definitions: Let $\vec{x}\in\mathcal{X}^{|D|}$ be a vector listing the first components (i.e., the unlabeled datapoints) for all members of $D$, in arbitrary order. For an index set $\mathcal{I}\subseteq\{1,\dots,|D|\}$, let $\vec{x}_{\mathcal{I}}\in\mathcal{X}^{|\mathcal{I}|}$ be the restriction of $\vec{x}$ to the indices $\mathcal{I}$. For a hypothesis $h$ and a vector $\vec{v}$, let $h(\vec{v})$ be the vector that results from applying $h$ element-wise to the entries of $\vec{v}$. Since the codomain of the hypothesis class is binary, i.e. $|\mathcal{Y}|=2$, we will assume here that $\mathcal{Y}=\{-1,1\}$ without loss of generality. For any $t\in\mathbb{N}$, let $\Sigma_{\mathit{trans}}^{t}$ denote the probability distribution over vectors of length $t$, whose entries are drawn i.i.d.\ and are equal to $-1$ with probability $\frac{m\,(|D|-m)}{|D|^{2}}$, equal to $1$ with probability $\frac{m\,(|D|-m)}{|D|^{2}}$, and are 0 otherwise. Furthermore, let $\Sigma_{\mathit{ind}}^{t}$ denote the probability distribution over vectors of length $t$ whose entries are independently drawn and $-1$ or $1$ with equal probability. Finally, denote by $\mathcal{B}$ the probability distribution over subsets of $\{1,\dots,|D|\}$ in which each element is contained in the subset independently with probability $2\,\frac{m\,(|D|-m)}{|D|^{2}}$.

In this notation, \citet[Def. 1 and p. 6]{EP09b} define the transductive Rademacher complexity $R_{\mathit{trans}}(\mathcal{H})$ as 
\[
({\textstyle {\frac{1}{m}+\frac{1}{|D|-m}})\cdot\mathbb{E}_{\vec{\sigma}\sim\Sigma_{\mathit{trans}}^{|D|}}\sup{}_{h\in\mathcal{H}}\;\vec{\sigma}^{T}h(\vec{x}).}
\]
Note that we can draw $\vec{\sigma}$ from $\Sigma_{\mathit{trans}}^{|D|}$ in two steps: we first draw the set of indices $\mathcal{I}$ from $\mathcal{B}$ whose entries in $\vec{\sigma}$ are nonzero, and then set $\vec{\sigma}$'s coordinates in $\mathcal{I}$ to $-1$ or $1$ with equal probability. Therefore, we can equivalently write 
\[
R_{\mathit{trans}}(\mathcal{H})=({\textstyle {\frac{1}{m}+\frac{1}{|D|-m}})\cdot\mathbb{E}_{\mathcal{I}\sim\mathcal{B}}\;\mathbb{E}_{\vec{\sigma}\sim\Sigma_{\mathit{ind}}^{|\mathcal{I}|}}\sup{}_{h\in\mathcal{H}}\;\vec{\sigma}^{T}h(\vec{x}_{\mathcal{I}}).}
\]
By \citet[Lemma 4 \& Thm. 6]{BM02},
$\mathbb{E}_{\vec{\sigma}\sim\Sigma_{\mathit{ind}}^{|\mathcal{I}|}}\sup{}_{h\in\mathcal{H}}\,\vec{\sigma}^{T}h(\vec{x}_{\mathcal{I}})\leq c_{1}\sqrt{d|\mathcal{I}|}$ for some absolute constant $c_{1}$. Thus, we can bound 
\begin{align}
R_{\mathit{trans}}(\mathcal{H}) & \leq c_{1}\,({\textstyle {\frac{1}{m}+\frac{1}{|D|-m}})\cdot\mathbb{E}_{\mathcal{I}\sim\mathcal{B}}\sqrt{d|\mathcal{I}|}\notag}\\
 & \leq c_{1}\,{\textstyle \frac{2}{m}\cdot\mathbb{E}_{\mathcal{I}\sim\mathcal{B}}\sqrt{d|\mathcal{I}|}\tag*{\ensuremath{m\leq|D|-m}}}\\
 & \leq\frac{2\,c_{1}\,\sqrt{d}}{m}\cdot\mathbb{E}_{t\sim\mathit{Binomial}\left(|D|,\frac{2\,m\,(|D|-m)}{|D|^{2}}\right)}\sqrt{t}\nonumber \\
 & \leq\frac{2\,c_{1}\,\sqrt{d}}{m}\,{\textstyle \left(\sqrt{6\,\frac{m\,(|D|-m)}{|D|}}+\mathbb{P}\left[\mathit{Binomial}\left(|D|,\frac{2\,m\,(|D|-m)}{|D|^{2}}\right)>6\,\frac{m\,(|D|-m)}{|D|}\right]\cdot\sqrt{|D|}\right)\notag}\\
 & \leq\frac{2\,c_{1}\,\sqrt{d}}{m}\,{\textstyle \left(\sqrt{6\,\frac{m\,(|D|-m)}{|D|}}+\exp(-2\,\frac{m\,(|D|-m)}{|D|})\cdot\sqrt{|D|}\right)\tag*{(Chernoff bound)}}\\
 & \leq\frac{2\,c_{1}\,\sqrt{d}}{m}\,{\textstyle \left(\sqrt{6\,m}+\exp\left(\frac{\ln|D|}{2}-m\right)\right)\tag*{(\ensuremath{1/2\leq\frac{|D|-m}{|D|}\leq1})}}\\
 & \leq\frac{2\,c_{1}\,\sqrt{d}}{m}\,{\textstyle \left(\sqrt{6\,m}+\exp\left(\frac{\ln m^{2}/\delta}{2}-m\right)\right)\tag*{(\ensuremath{|D|<m^{2}/\delta})}}\\
 & =\frac{2\,c_{1}\,\sqrt{d}}{m}\,{\textstyle \left(\sqrt{6\,m}+\exp\left(\frac{\ln1/\delta}{2}+\ln m-m\right)\right)\notag}\\
 & \leq\frac{2\,c_{1}\,\sqrt{d}}{m}\,{\textstyle \left(\sqrt{6\,m}+\exp\left(\frac{\ln1/\delta}{2}-(1-1/e)m\right)\right)\tag*{(\ensuremath{x-\ln x\geq(1-1/e)\,x}))}}\\
\intertext{By choosing a large enough constant in \cref{eq:mbound}, we can ensure that \ensuremath{(1-1/e)m\geq\frac{\ln1/\delta}{2}}. Then, we can continue:} & \leq\frac{2\,c_{1}\,\sqrt{d}}{m}\,{\textstyle \left(\sqrt{6\,m}+e^{0}\right)\leq\frac{2\,c_{1}\,\sqrt{d}}{m}\,(\sqrt{6}+1)\sqrt{m}\notag}\\
 & \leq\frac{c_{2}\,\sqrt{d}}{\sqrt{m}},\label{eq:rademacherbound}
\end{align}
where we set $c_{2}\coloneqq2\,(\sqrt{6}+1)\,c_{1}$. Putting together \cref{eq:twosideloss,eq:slackbound,eq:rademacherbound}, we obtain that, for all $h\in\mathcal{H}$, 
\[
{\textstyle L_{D\setminus S}(h)\leq L_{S}(h)+\alpha\quad\text{and}\quad L_{S}(h)\leq L_{D\setminus S}(h)+\alpha}
\]
where we defined 
\[
\alpha \coloneqq \frac{10.10+\sqrt{8/3\,\ln1/\delta}+c_{2}\,\sqrt{d}}{\sqrt{m}}.
\]
We have 
\begin{align*}
L_{D}(h) & =\frac{m}{|D|}L_{S}(h)+\frac{|D|-m}{|D|}L_{D\setminus S}(h)\\
 & \le\frac{m}{|D|}L_{S}(h)+\frac{|D|-m}{|D|}\left(L_{S}(h)+\alpha\right)\\
 & \le L_{S}(h)+\alpha
\end{align*}
and using a similar argument for the the other side, we obtain an error bound that holds uniformly across all {hypotheses}
\[
\left|L_{S}(h)-L_{D}(h)\right|\le\alpha\quad\forall h.
\]
Finally, we compute also a bound for the empirical risk minimizer. We set $h^{*}\coloneqq\argmin_{h\in\mathcal{H}}L_{D}(h)$. Then, we bound 
\begin{align*}
 & L_{D}(\hat{h})-L_{D}(h^{*})\\
={} & \frac{m}{|D|}\,\left(L_{S}(\hat{h})-L_{S}(h^{*})\right)+\frac{|D|-m}{|D|}\,\left(L_{D\setminus S}(\hat{h})-L_{D\setminus S}(h^{*})\right)\\
\leq{} & {\textstyle \frac{m}{|D|}\,\left(L_{S}(\hat{h})-L_{S}(h^{*})\right)+\frac{|D|-m}{|D|}\,\left(L_{S}(\hat{h})-L_{S}(h^{*})+2\,\alpha\right)}\\
={} & {\textstyle \underbrace{L_{S}(\hat{h})-L_{S}(h^{*})}_{\text{\ensuremath{\leq0}, by definition of \ensuremath{\hat{h}}}}+2\,\frac{|D|-m}{|D|}\,\alpha}\\
\leq{} & 2\,\alpha.
\end{align*}
By choosing the constant in \cref{eq:mbound} large enough, we can ensure\footnote{We may assume without loss of generality that $d\geq1$, since, if $d=0$, $\mathcal{H}$ only contains a single classifier and the claim holds trivially. If $d\geq1$, we can upper bound the term $\frac{12\cdot10.10^{2}}{\epsilon^{2}}$ by a multiple of $\frac{d}{\epsilon^{2}}$.} that 
\begin{align*}
m & \geq\frac{4}{\epsilon^{2}}\cdot3\,(10.10^{2}+8/3\,\ln1/\delta+c_{2}^{2}d).\\
\intertext{By Cauchy's inequality, this implies that}m & \geq\frac{4}{\epsilon^{2}}\cdot(10.10+\sqrt{8/3\,\ln1/\delta}+c_{2}\sqrt{d})^{2},
\end{align*}
and, by rearranging, that 
\begin{align*}
\epsilon & \geq2\,\frac{10.10+\sqrt{8/3\ln1/\delta}+c_{2}\sqrt{d}}{\sqrt{m}}\\
 & =2\cdot\alpha.
\end{align*}
Thus, with probability at least $1-\delta$, $\epsilon\geq L_{D}(\hat{h})-L_{D}(h^{*})$, and $\epsilon\ge\left|L_{S}(h)-L_{D}(h)\right|~\forall h$, as claimed. 
\end{proof}

\begin{algorithm}[t]
\DontPrintSemicolon \makedemproc \label[demproc]{alg:cjrsamplingexplicit} 
\textbf{Inputs}: agents $N$, slate size $k$, VC dimension $d$, error probability $\delta$\;

    $n_{x}\gets16\,C\,k^{4}\,(d+\log(k/\delta))$ ($C$ is the constant from \cref{lem:pacwithoutreplacement})\;

\If{$n\le2\cdot n_{x}$}{$n_{x}\leftarrow n$}

    $\epsilon\gets\frac{1}{4k^{2}}$\;

    $\bar{r}_{x}\gets n_{x}\left(\frac{1}{k}-\epsilon\right)$\;

    $\bar{r}\gets n\left(\frac{1}{k}-2\epsilon\right)$\;

    $S_{1}\gets N$\;

    $W_{0}\gets\emptyset$\;

    \For{$j=1,2,\dots,k$}{

    $X_{j}\gets$ draw $n_{x}$ agents from $N$ without replacement\;

    $Y_{j}\gets X_{j}\cap S_{j}$\;

    $\alpha_{j}\gets\pg{\begin{cases}
    \genquery{Y_{j}}{\left\lceil \bar{r}_{x}\right\rceil } & \text{if \ensuremath{|Y_{j}|\ge\bar{r}_{x}}}\\
\text{some arbitrary \ensuremath{\alpha\in\mathcal{U}}} & \text{else}
\end{cases}}$\;\label{ln:defalphaj}

    $\theta_{j}\gets\sup\left\{ \theta \mid \left|\textsc{supp}\left(\alpha_{j},\theta|Y_{j}\right)\right|\ge\bar{r}_{x}\right\} $\;\label{ln:defthetaj}

    $W_{j}\gets W_{j-1}\cup\{\alpha_{j}\}$\;

$r_{j}\gets\pg{\begin{cases}
    \left\lceil \bar{r}\right\rceil  & \text{if \ensuremath{j\leq n-k\left\lfloor \bar{r}\right\rfloor }}\\
    \left\lfloor \bar{r}\right\rfloor  & \text{else}
\end{cases}}$

$T_{j}\gets$ the $r_{j}$ agents in $S_{j}$ with largest $\disquery{\cdot}{\alpha_{j}}$

    $S_{j+1}\gets S_{j}\setminus T_{j}$\; }

    \Return $W_{k}$\; \caption{Democratic Process for BJR with Size-Constrained Queries (more explicit version of Process~2).}
\end{algorithm}

\thmvccjr*
\begin{proof}
For convenience, we define $\textsc{supp}(\alpha,\theta|S) \coloneqq \left\{ i\in S \mid u_{i}(\alpha)\ge\theta\right\} $ to be the set of agents in $S$ who have utility at least $\theta$ for statement $\alpha$. Further, we define \cref{alg:cjrsamplingexplicit}, which is equivalent to Process~2 but whose more explicit notation makes it easier to refer to specific values of the variables in this proof. Note that we have 
\[
\genquery{S}{\lceil r\rceil}=\argmax_{\alpha\in\mathcal{U}}\sup\left\{ \theta \mid |\textsc{supp}(\alpha,\theta|S)|\ge r\right\}
\]
    and hence we can write $\alpha_{j}$ defined in \cref{ln:defalphaj} of \cref{alg:cjrsamplingexplicit} as 
\begin{equation}
    \alpha_{j}=\argmax_{\alpha\in\mathcal{U}}\sup\left\{ \theta \mid \left|\textsc{supp}\left(\alpha,\theta|Y_{j}\right)\right|\ge\bar{r}_{x}\right\}.\label{eq:alpha}
\end{equation}

\paragraph{Step 1.}
We start by showing that with probability at least $1-\delta$, we have

\begin{equation}
\left|\frac{1}{n_{x}}\left|\textsc{supp}\left(\alpha,\theta|Y_{j}\right)\right|-\frac{1}{n}\left|\textsc{supp}\left(\alpha,\theta|S_{j}\right)\right|\right|\le\epsilon\label{eq:epsilonbound}
\end{equation}
for all $\alpha\in\mathcal{U}$, $\theta\in\mathbb{R}$, and $1 \leq j \leq k$.
For convenience, we define the indicator function:

\[
f_{\alpha,\theta}(i) \coloneqq \mathbb{I}\left[u_{i}(\alpha)\ge\theta\right].
\]
We can now write:
\begin{align*}
\frac{1}{n}\left|\textsc{supp}\left(\alpha,\theta|S_{j}\right)\right| & =\frac{1}{n}\left|\left\{ i\in S_{j} \mid u_{i}(\alpha)\ge\theta\right\} \right|\\
 & =\frac{1}{n}\sum_{i\in N} \, \mathbb{I}\left[u_{i}(\alpha)\ge\theta\right]\mathbb{I}\left[i\in S_{j}\right]\\
 & =\frac{1}{n}\sum_{i\in N}f_{\alpha,\theta}(i) \, \mathbb{I}\left[i\in S_{j}\right]
\end{align*}
and similarly:
\begin{align*}
\frac{1}{n_{x}}\left|\textsc{supp}\left(\alpha,\theta|Y_{j}\right)\right| & =\frac{1}{n_{x}}\sum_{i\in N}f_{\alpha,\theta}(i) \, \mathbb{I}\left[i\in Y_{j}\right]\\
 & =\frac{1}{n_{x}}\sum_{i\in N}f_{\alpha,\theta}(i) \, \mathbb{I}\left[i\in X_{j}\cap S_{j}\right]\\
 & =\frac{1}{n_{x}}\sum_{i\in X_{j}}f_{\alpha,\theta}(i) \, \mathbb{I}\left[i\in S_{j}\right].
\end{align*}
To bound the difference between these two terms, we map them to the learning-theoretic setting from \cref{lem:pacwithoutreplacement} as follows:
Let the domain $\mathcal{X}$ be the set of agents $N$, and the labels $\mathcal{Y}$ be $\{0,1\}$.
The set of labeled datapoints is $D \coloneqq \{(i,0)\}_{i\in N}$, from which we draw the uniform sample $S \coloneqq \{(i,0)\}_{i\in X_{j}}$ without replacement, and the hypothesis class is: 
\begin{align*}
\text{\ensuremath{\mathcal{H}}} \coloneqq \left\{ f_{\alpha,\theta}(\cdot) \, \mathbb{I}\left[\cdot\in S_{j}\right] \mid \alpha\in\mathcal{U},\theta\in\mathbb{R}\right\} .
\end{align*}
Hence, each hypothesis can be identified with a pair $(\alpha,\theta)$ and it is then easy to see that the losses from \cref{lem:pacwithoutreplacement} are precisely the terms we are trying to relate:
\begin{align*}
L_{S}(\alpha,\theta) & =\frac{1}{n_{x}}\left|\textsc{supp}\left(\alpha,\theta|Y_{j}\right)\right| \text{ and} \\
L_{D}(\alpha,\theta) & =\frac{1}{n}\left|\textsc{supp}\left(\alpha,\theta|S_{j}\right)\right|.
\end{align*}
Hence, \cref{lem:pacwithoutreplacement}, along with a union bound across the $k$ steps, tells us that if the sample size satisfies: 
\begin{align}
n_{x} & \ge C \cdot \frac{\text{\textsc{vc-dim}}(\mathcal{H})+\log k/\delta}{\epsilon^{2}}\label{eq:nx}\\
 & =16 \, C \, k^{4}(\text{\textsc{vc-dim}}(\mathcal{H})+\log k/\delta),\nonumber
\end{align}
then \cref{eq:epsilonbound} holds with probability at least $1-\delta$. 
To show \cref{eq:epsilonbound}, it remains to relate $\text{\textsc{vc-dim}}(\mathcal{H})$ to the VC dimension $d$ of our statement space. Note that for all hypotheses in $\mathcal{H}$, all datapoints in $S_{j}$ are constrained to $0$ due to the factor $\mathcal{I}[\cdot \in S_j]$.
Compared to a definition without this indicator factor, this restriction does not increase the VC dimension of the hypothesis class since the datapoints in $S_j$ cannot be part of any shattered subset. Consequently, $\text{\textsc{vc-dim}}(\mathcal{H})$ is at most equal to the VC dimension of the hypothesis class 
\begin{align*}
\left\{ f_{\alpha,\theta}(\cdot) \mid \alpha\in\mathcal{U},\theta\in\mathbb{R}\right\} .
\end{align*}
It is easy to verify that the VC dimension of this set of indicator functions corresponds to our notion of VC dimension $d$, hence $\text{\textsc{vc-dim}}(\mathcal{H})\le d$, which means that our $n_{x}$ from \cref{alg:cjrsamplingexplicit} satisfies \cref{eq:nx} and therefore \cref{eq:epsilonbound} holds with the desired probability.

\paragraph{Step 2.}
Next, we show that, when \cref{eq:epsilonbound} holds, it must hold that, for each iteration $j$, all of the agents $T_{j}$ removed in this iteration have utility at least $\theta_{j}$ for the selected statement $\alpha_{j}$.
For this, it suffices to show that there are at least $r_{j}$ agents in $S_{j}$ with utility at least $\theta_{j}$ for $\alpha_{j}$, i.e., that $\left|\textsc{supp}\left(\alpha_{j},\theta_{j}|S_{j}\right)\right|\ge r_{j}$.
First, observe that we defined $r_{j}$ such that we always have $|S_{j}|\ge r_{j}$, since
\[ \sum_{1 \leq j \leq k} r_j \leq k \, \lfloor n \, (\tfrac{1}{k} - 2 \epsilon)\rfloor + \left(n - k \, \lfloor n \, (\tfrac{1}{k} - 2 \epsilon)\rfloor\right) \leq n.\]
Secondly, in the edge case where $|Y_{j}|<\bar{r}_{x}$, we have, by its definition in \cref{ln:defthetaj}, $\theta_{j}=-\infty$ and hence the requirement is trivially satisfied.
In the more interesting case of $|Y_{j}|\ge\bar{r}_{x}$, the same definition implies that:
\begin{align*}
    \left|\textsc{supp}\left(\alpha_{j},\theta_{j}|Y_{j}\right)\right|&\geq \bar{r}_{x}.
\end{align*}
By applying our assumption of \cref{eq:epsilonbound}, it follows that:
\begin{align*}
    \frac{1}{n}\left|\textsc{supp}\left(\alpha_{j},\theta_{j}|S_{j}\right)\right|+\epsilon & \ge\frac{\bar{r}_{x}}{n_{x}}\\
    \intertext{and thus that}
    \left|\textsc{supp}\left(\alpha_{j},\theta_{j}|S_{j}\right)\right| & \ge n\cdot\left(\frac{\bar{r}_{x}}{n_{x}}-\epsilon\right)\\
     & =\bar{r}.
\end{align*}
    Since the left-hand-side is an integer and $r_{j}\le\left\lceil \bar{r}\right\rceil $, it follows that
\begin{equation}
\left|\textsc{supp}\left(\alpha_{j},\theta_{j}|S_{j}\right)\right|\ge r_{j}
\end{equation}
as desired.

\paragraph{Step 3.}
We can now finally show that the algorithm satisfies BJR. Let the matching $\omega$ be such that, for all rounds $j\in\{1,\dots,k\}$ and agents $i\in T_{j}$, we have $\omega(i)=\alpha_{j}$. Note that any two $T_{j},T_{j'}$ differ in size by at most $1$, hence clearly the balancing condition (i.e., $|\{i:\omega(i)=w\}|\in\{\lceil n/k\rceil,\left\lfloor n/k\right\rfloor \}$ for all $w\in W_{k}$) can be satisfied by assigning the remaining agents in $S_{k+1}$ appropriately to statements in $W_{k}$. Having defined a balanced matching $\omega$, consider a coalition $S\subseteq N$ of size $\ge n/k$, a candidate $\alpha\in\mathcal{U}$, and a $\theta\in\mathbb{R}$ such that $u_{i}(\alpha)\ge\theta$ for all $i\in S$. 

The number of agents remaining after the $k$ iterations satisfies $|S_{k+1}|<n/k$, hence $S\nsubseteq S_{k+1}$. To see this, consider the number of agents, $r_{j}$, removed in each round. During 
\[
\max\left\{ \min\left\{ n-k\left\lfloor \bar{r}\right\rfloor ,k\right\} ,0\right\} 
\]
rounds, we remove $\left\lceil \bar{r}\right\rceil $ agents per round, and for the remaining rounds we remove $\left\lfloor \bar{r}\right\rfloor $ agents per round. It follows that in average, we remove $\min\left\{ \frac{n}{k},\left\lceil \bar{r}\right\rceil \right\} $ agents per round. It is easy to verify that $\min\left\{ \frac{n}{k},\left\lceil \bar{r}\right\rceil \right\} \ge\bar{r}$, hence 
\[
    |S_{k+1}|\le n-k\bar{r}=2\cdot k\cdot n\cdot\epsilon=\frac{n}{2k}.
\]
This means that for some iteration $q\in[k]$ we have $S\cap T_{q}\neq\emptyset$. Let $q$ be the iteration where this happens the first time, which implies that $S\subseteq S_{q}$ and thus that 
\begin{align*}
\frac{n}{k} & \le\left|\textsc{supp}\left(\alpha,\theta|S\right)\right|\\
 & \le\left|\textsc{supp}\left(\alpha,\theta|S_{q}\right)\right|,\\
 \intertext{or, equivalently, that}
\frac{1}{k} & \le\frac{1}{n}\left|\textsc{supp}\left(\alpha,\theta|S_{q}\right)\right|.
\end{align*}
Assuming \cref{eq:epsilonbound}, which holds with probability at least $1-\delta$ as established in the first step, it follows that 
\begin{align*}
\frac{1}{k}-\epsilon & \le\frac{1}{n_{x}}\left|\textsc{supp}\left(\alpha,\theta|Y_{q}\right)\right|,\\
\intertext{or, equivalently, that }
    \bar{r}_{x} & \le\left|\textsc{supp}\left(\alpha,\theta|Y_{q}\right)\right|.
\end{align*}
Hence, $\alpha$ is a candidate in the definition of $\alpha_{q}$ as expressed in \cref{eq:alpha}.
Therefore, it must be that $\theta_{q}\ge\theta$. As shown in the second step, all agents in $i\in T_{q}$ have utility $u_{i}(\alpha_{q})\ge\theta_{q}\ge\theta$. Since at least one agent $i\in S$ is in $T_{q}$, we have $\theta\le u_{i}(\alpha_{q})=u_{i}(\omega(i))$, which means that there can be no violation of BJR.
\end{proof}

\section{Deferred Details About Pilot}
\label{app:experiments}
\subsection{Representativeness of the Samples}
\label{app:demographics}

\begin{figure}[htb]
\centering
\includegraphics[width=\linewidth]{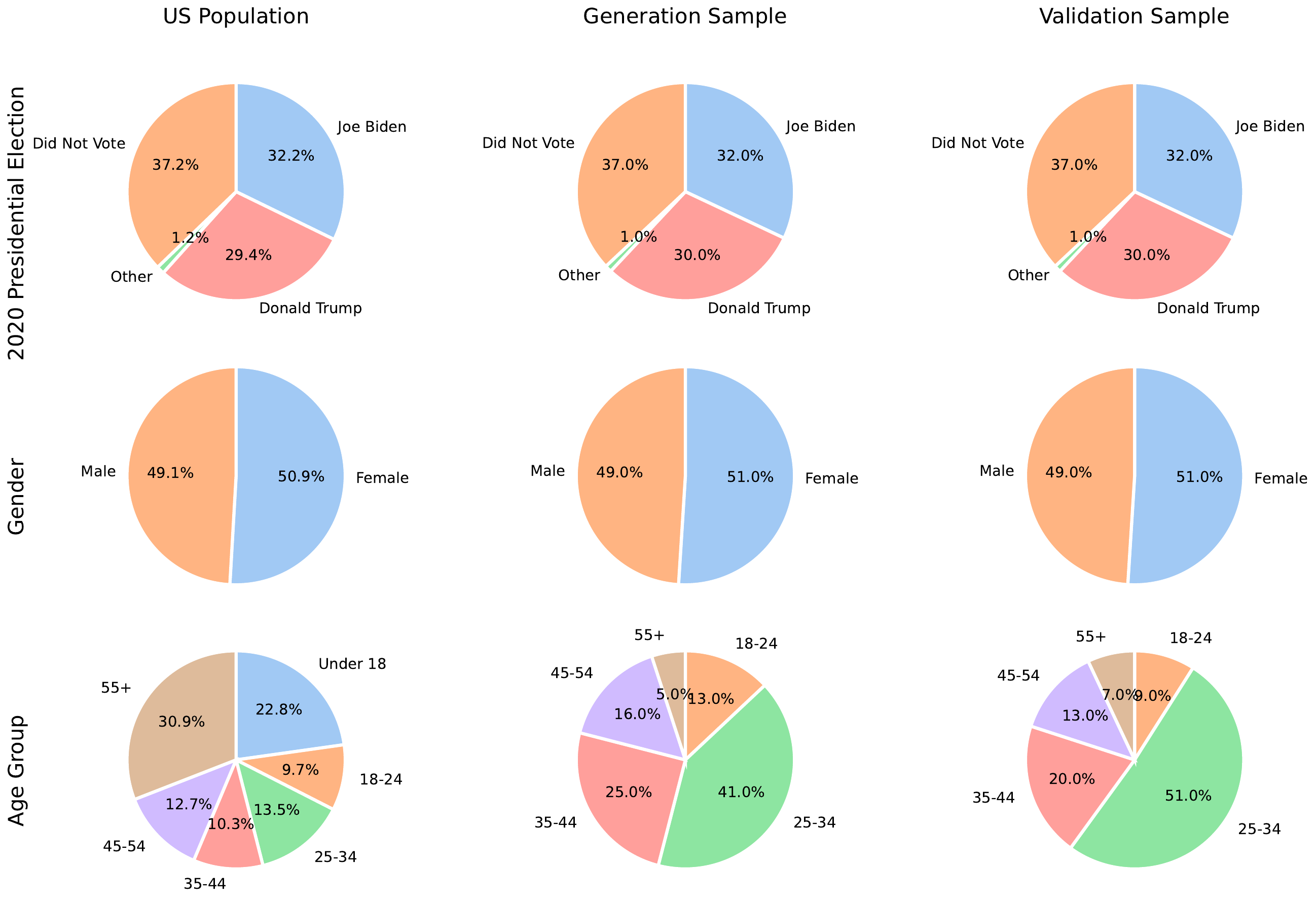}
\caption{Demographic composition of both samples, compared to the US population as of the 2020 Census \cite{USCensus2020} and the 2020 U.S. Presidential Election results \cite{ElectionResults2020}.}
\label{fig:demographics}
\end{figure}

{
We recruited a sample of U.S.\ residents on Prolific, stratified on gender and voting behavior in the 2020 U.S.\ presidential election (Joe Biden, Donald Trump, third-party candidate, and nonvoter).
We adopt these two criteria of stratification because they are especially salient for the topic of abortion policy.
We excluded participants below the age of 22 at the time of the survey, to ensure that all respondents were of voting age in the 2020 election.
As shown in \cref{fig:demographics}, our sample mirrors the U.S.\ population's composition in terms of these stratified features.
Our sample is not just representative along gender and voting-behavior lines but also within intersections of these features.
Compared to the U.S.\ population, our sample skews young; in particular, the 25--34 and 35--44 year cohorts are overrepresented.
}

{To maintain data integrity, we filtered out submissions suspected to be generated by language models or those that appeared extremely low effort.
Specifically, we recruited, for each of our surveys, 110 participants on Prolific, and manually identified responses that exhibited patterns typical of AI-generated content or minimal engagement with the survey questions.\footnote{For example, one participant’s response included the phrase, “You’ve hit the Free plan limit for GPT-4o”, and another copy-pasted the same answer repeatedly.}
For the generation sample, we excluded 4 of the 110 submissions for LLM usage and 2 submissions for extremely low effort; and for the validation sample, we excluded 5 submissions for LLM usage and 2 for extremely low effort.
A member of our team then selected the 100 submissions for our survey from those not filtered out, by re-establishing proportionality along voting behavior, gender, and intersections of the two categories.
To avoid bias, the team member was blinded to the content of submissions for this decision.
}

\subsection{Slates}
\label{app:slate}
\definecolor{DataControlColor}{RGB}{23,150,150}
\definecolor{AccuracyColor}{RGB}{124,205,124}
\definecolor{PrivacySecurityColor}{RGB}{23,88,200} 

{
\subsubsection{Slate Generated by Our Democratic Process}
\label{app:ourslate}
\begin{enumerate}
\item[S1.] I believe that abortion should be legal and that a woman has the right to choose what she does with her body. It is crucial for the decision to have an abortion to remain a healthcare decision between a woman and her doctor, free from political interference. Society should respect and uphold these rights to ensure women's autonomy and wellbeing.
\item[S2.] I believe that abortion should be legal and accessible to everyone, as it is essential healthcare. It is critical for women to have control over their own bodies and make decisions about abortion without government interference. Safe and legal access to abortion ensures that women can make the best choices for their health and well-being.
\item[S3.] I believe that abortion should be illegal in most cases and only allowed under specific circumstances. It is morally wrong and equates to taking a life, as I believe in the sanctity of human life starting from conception. People should take responsibility for their actions and should never use abortion as a form of birth control.
\item[S4.] I believe that abortion should be legal and that women should have the right to make decisions about their own bodies. The decision to have an abortion should be a private matter, free from societal judgment, as it's a crucial aspect of women's health and safety. Society should ensure that abortions are safe, legal, and accessible, respecting the individual's autonomy and privacy in making such decisions.
\item[S5.] I believe that abortion should not be used as a form of birth control. However, it should be legally permissible in cases where the mother's life is in danger, and in instances of rape or incest within an early timeframe. Furthermore, there should be robust educational efforts to prevent unwanted pregnancies and ensure people are informed about the complexities and consequences of abortion.
\end{enumerate}
Statements S1 and S5 were generated through nearest-neighbor clustering (with clusters of size 10), whereas the other three statements were generated by balanced k-means clustering.
Statements S1, S2, and S5 were generated in the iteration of the greedy algorithm in which they were selected (i.e., the first, second, and fifth iteration, respectively), whereas S3 was generated in the first iteration and S4 in the third iteration.

Three of the statements (S1, S2, and S4) above are very similar. They express a position that strongly favors broadly legal abortions and sees abortion as a private choice, in the realm of healthcare, and as part of a right to bodily autonomy.
Statement S3 wants abortion to be broadly illegal, albeit with some exceptions, since it sees abortion as murder of a human being with full moral consideration. It also expresses concerns about an abuse of abortion as birth control.
Statement S5 shares this concern about abuse, and expresses discomfort with abortion, but takes a more moderate position. Abortion should be legal in medical emergencies, or cases of rape and incest; the statement does not take a stance on whether abortion should be legal or illegal outside of these circumstances. Instead, the statement hopes to reduce the prevalence of abortion through education.

To check that this slate plausibly represents the generation sample, we coarsely cluster the participants by which of the five example statements (full statements in \cref{app:generationsurvey}) they agree with most.
Specifically, we match each participant in the generation sample to the example statement they rated most highly, matching them fractionally if several statements are tied for the highest rating.
In this case, the number of participants matched to each statement is as follows:
\begin{center}
\small
\begin{tabular}{rl}
\toprule
\# matched & statement \\
\midrule
32.5 & I think abortion should be a personal decision between a woman and her doctor\dots \\
27.2 & I believe that abortion should be legal and accessible because women have the right\dots \\
17.0 & I think abortion should be restricted to certain circumstances\dots \\
15.0 & I believe that abortion should be illegal because it involves taking a human life\dots \\
8.3  &  believe that abortion should be allowed in the first trimester but restricted afterward\dots \\
\bottomrule
\end{tabular}
\end{center}

Conceptually, the three pro-choice statements on our slate (S1, S2, and S4) are a blend of the first two statements (minus the call for contraception and sex education), so representing $32.5+27.2=59.7$ participants by three statements (which ideally should represent $3 \, n/k = 60$ participants) seems quite accurate.
These two example statements are also the pair whose ratings are the most correlated (correlation coefficient 0.88), so it is plausible to treat the agents whose favorite is among those two as a single bloc.

The pro-life statement on the slate (S3) best matches the fourth statement above, given that both statements express that abortion should be illegal and equate it with murder (though the specifics vary).
Representing 15 pro-life participants by one statement is also plausible.

The moderate statement on our slate (S5) most closely resembles the third statement in the table, given that both statments want abortion to be legal in cases of incest, rape, and risk to life, and given that both look for tools to reduce the prevalence of abortions besides of criminalization (though, again, some specifics vary).
Representing these 17 participants by a statement also seems close to proportional.

This leaves the 8.3 adherents of the last statement (advocating for a concrete temporal cutoff for legality) without an obvious representative on our slate, but this group are too small to cause a violation of BJR on their own.
The fact that agreement with this statement does not correlate much with agreement to the other statements\footnote{{The two correlations with largest magnitude are positive correlations with the pro-choice statements, with correlation coefficients of 28\% and 24\%.}} further suggests that the about 8 adherents of the last statement are not the core of a cohesive group whose BJR guarantee would be violated.

If we repeat the above assignment to favorite statements, but only among the second, third, and fourth statement in the table above (i.e., one pro-choice statement, one hesitant intermediate statement, and one pro-life statement), the respective numbers of matched participants are 62, 21.5, and 16.5.
Once more, this suggests that representing the first of these three viewpoints by three statements, and the latter two viewpoints by one statement each, is proportional to the opinions in the generation sample.

\subsubsection{Baseline Slate Generated By GPT-4o}
\label{app:baselineslate}
\begin{enumerate}
\item[B1.] I believe abortion should be legal because it is a woman's right to choose what happens with her body. It is important for maintaining reproductive freedom and ensuring safe and accessible healthcare options. Women need to make decisions that are best for them and their families without government interference.
\item[B2.] In my opinion, abortion should be illegal except in cases where the mother's life is at risk or in instances of rape or incest. This stance protects unborn children while allowing necessary exceptions. It's about balancing moral concerns with compassion for those in difficult situations.
\item[B3.] Abortion should be legal up to a certain point in the pregnancy, such as the first or second trimester. After that, I think it needs to be restricted unless there are severe medical reasons. Society needs to find a middle ground that respects both a woman's right to choose and the potential life of the fetus.
\item[B4.] I think abortion should be legal and accessible at all stages of pregnancy. Women need to have control over their reproductive health for various personal and medical reasons. Restrictions can force women into unsafe or undesirable circumstances.
\item[B5.] Abortion should be illegal in all cases, as I believe it is morally wrong and equivalent to ending a potential life. Society should look for alternative solutions such as adoption and support for pregnant women. We need to value and protect all human life, starting from conception.
\end{enumerate}
}

\subsubsection{Distribution of Ratings in Validation Sample}
\label{app:ratingdistribution}
{
Below, we show the distribution of ratings for the statements of both slates in the validation sample.
For each slate, we disaggregate agents depending on which statement they are assigned to, and show, for each group and each statement, a histogram of rating levels.

First, we show the distribution of rating levels for our slate:
\begin{center}
\includegraphics[width=.8\textwidth]{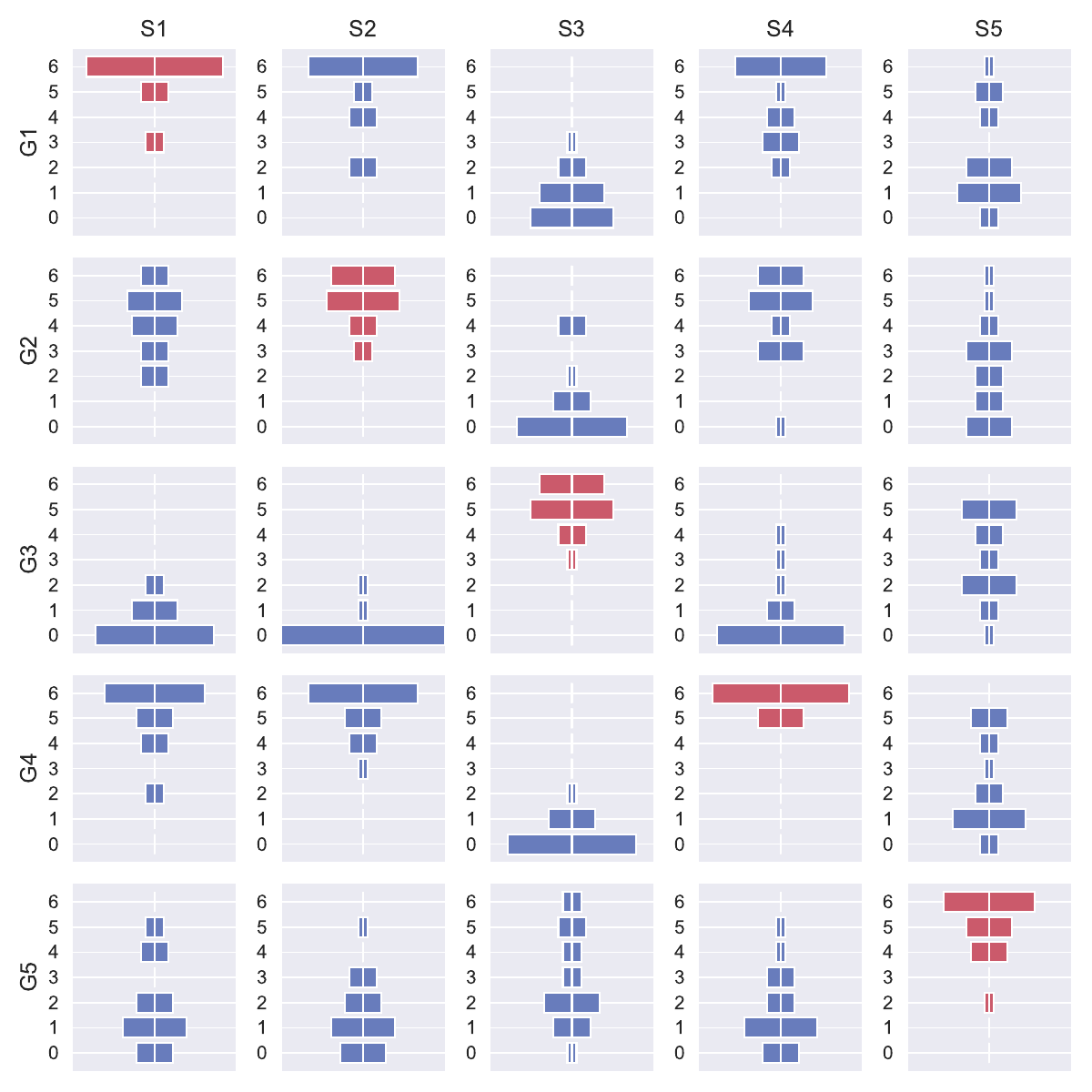}
\end{center}
Each row corresponds to a group; for example, G3 represents the 20 participants assigned to statement S3. For each group, we polot the frequencies of rating levels given by members of this group to statements S1 through S5.

We observe that rating levels on the diagonal, (i.e., ratings for the groups' assigned statement) are higher than those off diagonal (i.e., ratings for other statements).
We also see that there is substantial disagreement in the dataset; for example, group G3 (the 20 agents matched to the pro-life statement) are very critical of the three pro-choice statements S1, S2, and S4.
We generally see that participants of all groups rate these three pro-choice statements similarly, attesting to their striking similarity.

We now show the corresponding figure for the baseline slate, where groups G1 through G5 now refer to the agents assigned to B1 through B5, respectively:
\begin{center}
\includegraphics[width=.8\textwidth]{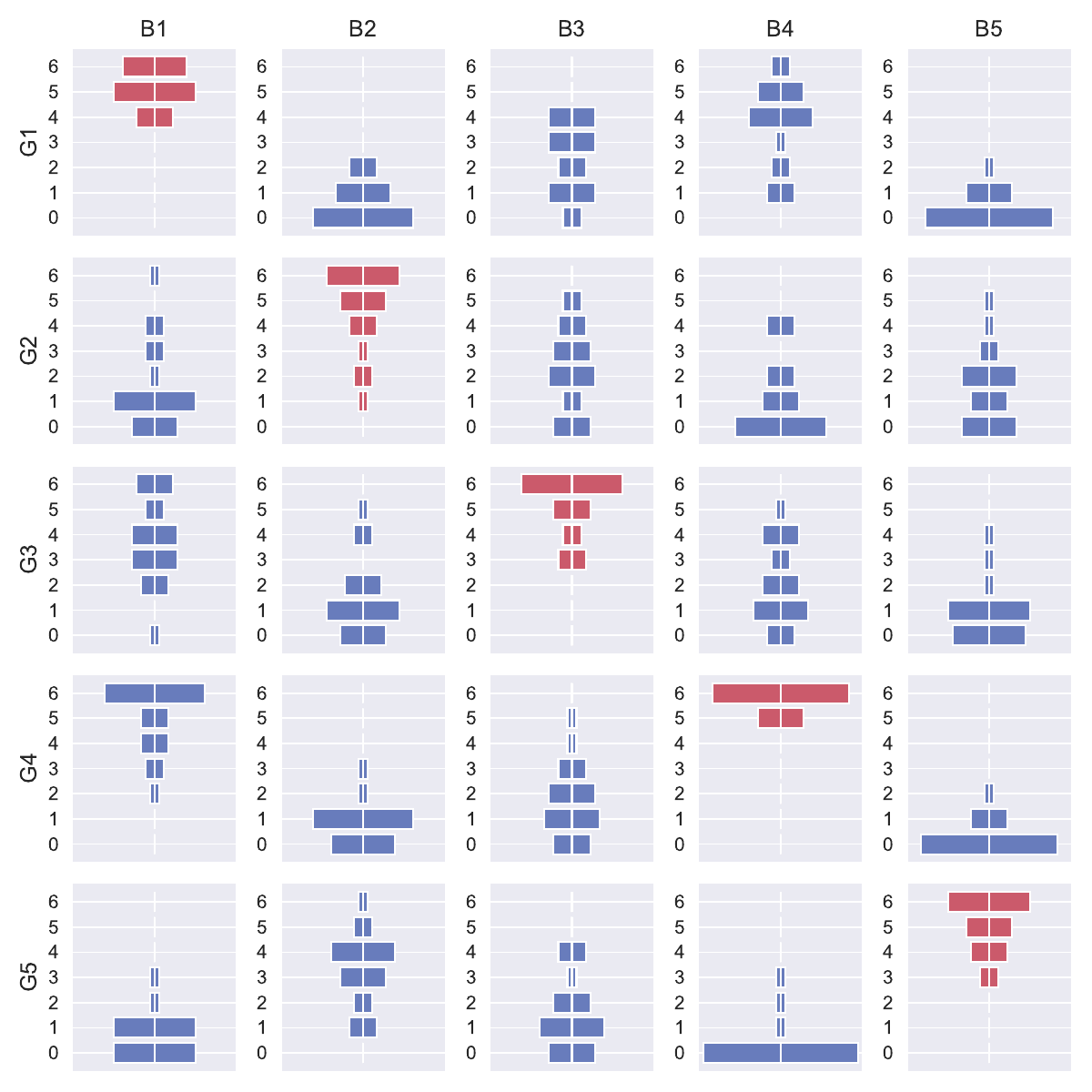}
\end{center}

We again see higher ratings on the diagonal than outside of it.
We can also see that statement B2 seems slightly less well received by its assigned group than the other statements on this slate or the five statements on our slate.
}

\section{Prompts}\label{app:prompt-design-process}
The formatting of the prompts has been lightly edited for readability. 
\subsection{Summarization}\label{app:summarization}

\subsubsection*{System Prompt}
\begin{prompt}
You will be provided with a user's response to a survey, in which they describe their opinions on a topic in detail. Your task is to produce a detailed summary of that user's opinion. Your response should be in JSON according to the format specified below. 
\end{prompt}

\subsubsection*{Prompt}
\begin{prompt}
\textit{User survey responses:}
\\
\{user\_data\}
\\ \\
\textit{Output instructions:}
\\
Complete the following entries:

\begin{itemize}
\item most\_important\_aspects (List[str]): List the aspects of the topic that are most important to the user. Each aspect should be succinct and self-contained, so that it can be understood without context. (For example, instead of writing "Religious beliefs", write "Believes X due to religious beliefs") 
\item specific\_details (List[str]): List any specific details or examples the user provided to support their opinion. Each detail should be succinct and self-contained, so that it can be understood without context. These details, together with the\\most\_important\_aspects above, should be enough to mostly reconstruct the user's opinion.
\item user\_background (List[str]): List any personal information the user may have divulged that is relevant to their opinion.
\item overall\_summary (str): Write a detailed 2-3 sentence summary of the user's opinion, taking into account all the context provided above.
\end{itemize}

\noindent Respond in JSON with the fields above filled in.
\end{prompt}

\subsection{Tagging}

\subsubsection*{System Prompt}
\begin{prompt}
You will be provided with a user's response to a survey, in which they describe their opinions on a topic in detail. Then you will be provided with a list of aspects of that topic. For each aspect, your task is to rate the extent to which it pertains to the user's response or captures the user's opinion. The scale is as follows:

\begin{enumerate}
\item Strongly goes against user's opinion 
\item Goes against user's opinion
\item Somewhat goes against user's opinion
\item Neutral / unknown 
\item Somewhat aligned with user's opinion 
\item Aligned with user's opinion
\item Strongly aligned with user's opinion
\end{enumerate}

\noindent Respond in JSON, mapping each aspect to your rating. It is important to copy the EXACT wording of the aspect with no changes.
\end{prompt}

\subsubsection*{Prompt Template}
\begin{prompt}
\textit{User survey responses:}
\\
\{user\_data\}
\\ \\
\textit{Output instructions:}
\\For each of the following aspects, rate on the scale from 1-7 how well it pertains to the user's response or captures the user's opinion. 
\\
\{list\_of\_fields\}
\\ \\
Respond in JSON, with the above fields as keys and your ratings as values.
\end{prompt}

Example fields used include: \texttt{most\_important\_aspects} and \texttt{specific\_examples}.

\subsection{Discriminative Query}\label{app:discriminative_query}
\subsubsection*{System Prompt}
\begin{prompt}
You are an AI-based text completion system, and you are tasked with helping a user fill out an in-depth opinion survey.\\ \\
You will receive in your input a sequence of questions along with the answers already given by the user. Then, you will receive the question that the user will answer next. Based on the question and what you have learned so far about the user, you must make your best guess what the user will answer. If you guess right, the user can save a lot of effort and time, so it is important that you suggest what the user would be most likely to answer. Your prediction must strictly adhere to the format at the end of the prompt.
\end{prompt}

\subsubsection*{Prompt Template}
\begin{prompt}
Question 1:
\\ \{question\}\\ \\
Answer 1: <ANSWER FORMAT: string.>
\\\{freetext\_answer\}\\ \\
...\\ \\
Question 6:\\ \{question\} \\ \\
Answer 6: Part A — User's level of agreement: <ANSWER FORMAT: numeric score and text label. Must be exactly equal to one of the following options: "0 = very poorly", "1 = poorly", "2 = moderately", "3 = well", "4 = very well", "5 = excellently", or "6 = exceptionally".>\\ \{choice\_numeric\} = \{choice\}\\ \\
Answer 6 Part B — User's explanation: <ANSWER FORMAT: string>\\ \{explanation\} \\ \\
... \\ \\
Question 11: \\ \\ To what extent does this statement capture your full opinion on abortion? \\ \\ \{statement\} \\ \\ You must now infer the user's most likely agreement level with Question 11 from their previous responses. Your response must exactly equal one of the following options: "0 = very poorly", "1 = poorly", "2 = moderately", "3 = well", "4 = very well", "5 = excellently", or "6 = exceptionally"
\end{prompt}

\subsection{Generative Query}\label{app:generative_query}
\subsubsection*{System Prompt}
\begin{prompt}
You will be provided with a list of users and their opinions on a topic. Your task is to write a paragraph all of the users would agree with. Specifically, your measure of success is based on how satisfied the *least satisfied* user is with the paragraph.

When selecting the content for your paragraph, follow these guidelines:
\begin{itemize}
\item For any aspect of the topic that all users agree on, include that detail in your paragraph. Especially when the users disagree on major aspects of the topic, it's important to make note of any minor aspects that they all agree on and include them.
\item Sometimes there is an aspect of the topic that most users agree on, but some users don't express a strong opinion on. If it's reasonable to assume that those users would also agree, you should still include that detail in your paragraph. 
\item For any aspect of the topic that the users fundamentally disagree on, omit details from the paragraph. Only include this aspect if there is a way to phrase it that everybody would agree with. 
\end{itemize}
As for writing style, your paragraph be written like an answer to the following survey question:\\ \\
> "\{statement\_question\_text\}" \\ \\
When writing, you should first think about the content you want to include (cf. the instructions above on content), and then conjure an imaginary new user who holds those precise beliefs. Then, when writing your paragraph, you should answer the survey question as if you were that user. Your paragraph should never reference the other users -- you should write "I think..." statements instead of "Some users think..." statements. Your paragraph should not even implicitly mention the users -- you should never write "Opinions vary about...". The users' opinions are only there to help you understand the spread of opinions for the content selection step -- when it comes to the writing step, you should write from the perspective of a single (imaginary) new user.
\end{prompt}

\subsubsection*{Prompt Template}
\begin{prompt}
\textit{List of users:}
\\ \\
\{user\_descriptions\}
\\ \\
\textit{Instructions:}
\\ \\
This task consists of two main parts. First, you need to determine the content of your paragraph (recall the content guidelines from the beginning of the instructions). Second, you need to write the paragraph itself, which should be phrased like a single user's answer to the survey question (recall the writing guidelines from the beginning of the instructions).
\\ \\
Step 1. \textit{Identify common themes.} List every viewpoint that is expressed by multiple users at once. For each such viewpoint, thoughtfully make note of how many users hold that viewpoint. Keep in mind that some users may express multiple viewpoints.
\\ \\
Step 2. \textit{Identify key disagreements.} Identify all aspects of the topic where users have differing opinions. For each such aspect, thoughtfully make note of how many users hold each opinion (or if there are users who abstain).
\\ \\
Step 3. \textit{Analyze potential contents of paragraph.} For each aspect, discuss how (if at all) it should be included in the paragraph. 
\\ \\
Step 4. \textit{Write the paragraph.} This is the most important step, for which you need to do the most thinking. Write a paragraph that captures the opinions of the users as best as possible. 
\\ \\
When doing so, keep in mind these *content guidelines*:
\begin{itemize}
\item Your measure of success is based on how satisfied the *least satisfied* user is with the paragraph.
\item When all users agree on an aspect, include that aspect. (Still include if some users don't express a strong opinion on it.)
\item When users disagree on an aspect, omit that aspect from the paragraph. Focus on the aspects that all users can agree on.
And additionally keep in mind these *writing style guidelines*:
\item Your paragraph should sound like a single (imaginary) user's to the survey question.
\item Your paragraph should neither explicitly nor implicitly mention the users whose data you are working with. That data should only be used to determine the content -- it should not be referenced in the writing itself.
\end{itemize}

\noindent \textit{Output instructions:}
\\ \\ 
Respond in JSON as follows: 
\{\{

"step1" : <your response to step 1>,

"step2" : <your response to step 2>,

"step3" : <your response to step 3>,

"step4" : <your response to step 4> (str),

\}\}
\end{prompt}

\subsection{Example Statements} \label{app:zero-shot}
\subsubsection*{System Prompt}
For consistency with the baseline prompts, we passed the following prompt to GPT-4o as a system prompt (with empty standard prompt):
\begin{prompt}
Write in bullet points five different beliefs that people from the US might have about abortion. Each belief should be written such that as many people as possible would agree with it. They should sound like sentences a real person would say, written in the first person. They should be written like answers to the following survey question:
Summarize your position on abortion in your own words.
Please write \textit{how you think society should deal with abortions}, and \textbf{give reasons} that support this policy.
\\ \\
Please write \textit{exactly 3 sentences}: Be as precise as possible and prioritize the points that are most important to you.\\
Your answer should be self-contained, which means that you can repeat things you already wrote as well as make new points.
\end{prompt}

The statements generated by this prompt can be found as part of the survey questions in \cref{app:generationsurvey}.

\subsection{Baseline Prompts} \label{app:one-shot}
\subsubsection*{System Prompt}
\begin{prompt}
Write in bullet points five different beliefs that people from the US might have about abortion. Each belief should be written such that as many people as possible would agree with it. They should sound like sentences a real person would say, written in the first person. They should be written like answers to the following survey question:\\
Summarize your position on abortion in your own words.\\
Please write \textit{how you think society should deal with abortions}, and \textit{give reasons} that support this policy.
\\ \\
Please write \textit{exactly 3 sentences}: Be as precise as possible and prioritize the points that are most important to you.\\
Your answer should be self-contained, which means that you can repeat things you already wrote as well as make new points.
\\ \\ 
To help you better understand typical beliefs that people from the US have on abortion, attached are survey responses from a representative sample.
\end{prompt}

\subsubsection*{Prompt}

As the full prompt is extremely long, we include only the prompt template below. For the survey data, see \url{https://github.com/generative-social-choice/survey_data}.
\begin{prompt}
Survey responses from 100 people:\\ \\ 
Person 1\\
   \textit{(Person 1's complete survey responses)}
\\ \\ 
Person 2 \\
    \textit{(Person 2's complete survey responses)}
\\ \\
... \\ \\ 
Person 100 \\
    \textit{(Person 100's complete survey responses)}

\end{prompt}

\section{Survey Questions}
\label{app:survey}

Below are the full question prompts of the two Prolific surveys we ran.

\subsection{Generation Survey}
\label{app:generationsurvey}
\subsubsection*{Informed Consent}

We are a team of university researchers. We want to understand in detail what people like you think about social questions, in this case about abortion. We want to study the diversity of people’s beliefs, and whether algorithmic tools can help summarize and analyze these differing points of view.\\

\noindent\textit{What will I need to do and how long will the study last?}
We will ask you \textbf{10 questions}.
We expect that you will be in this research study for less than an hour. \\

\noindent\textit{Compensation: }
Your pay will \textbf{not} depend on your opinions: all good-faith responses will get fully compensated, so please just write what you think.
You may not use ChatGPT or other AI tools to write or edit your responses.
\textbf{We will award a \$1 bonus} if your submission satisfies all minimum length requirements across all questions.\\

\noindent\textit{Who will see your responses?}
The data you provide will be immediately anonymized. We may later on publish this anonymous data. By continuing this survey, you agree to this use of your responses.\\

\subsubsection*{Part 1/3: Background Questions}

We will begin by asking you \textbf{4 questions} touching on specific parts of your opinion.

\subsubsection*{Part 2/3: Your Position}

This part only has \textbf{one question}.
We will ask you to summarize, in three sentences, your overall position on abortion.

\subsubsection*{Part 3/3: Candidate Summaries}

To help researchers and policymakers understand public opinion on abortion, we will condense the survey responses into a short summary.
We will now present \textbf{5 candidate summaries} to you, and your task is to rate to what extent they would capture your point of view.

An ideal summary should

\begin{enumerate}
\item perfectly capture all your thoughts and feelings on abortion,
\item include the reasoning behind your opinion, and
\item be concrete and actionable.
\end{enumerate}

\subsubsection*{Background Questions: 1/4}

\textit{How often do think about abortion or discuss it with others?} How does this topic make you feel? \\

\noindent Please write \textbf{two} or more sentences.

\subsubsection*{Background Questions: 2/4}

\textit{Do you think abortion should be legal or illegal?} Which circumstances does your answer depend on? \\

\noindent Please write \textbf{two} or more sentences.

\subsubsection*{Background Questions: 3/4}

\textit{Where do your beliefs about abortion come from?} For example, did particular life experiences influence your beliefs? \\

\noindent Please write \textbf{two} or more sentences.

\subsubsection*{Background Questions: 4/4}

\textit{Can you describe a situation where you are not sure if abortion is appropriate or not?} If so, what makes this situation borderline or unclear? \\

\noindent Please write \textbf{two} or more sentences.

\subsubsection*{Your Position}

Summarize your position on abortion in your own words.\\

\noindent Please write \textbf{how you think society should deal with abortions}, and \textbf{give reasons} that support this policy.\\

\noindent Please write \textbf{exactly 3 sentences}: Be as precise as possible and prioritize the points that are most important to you.\\

\noindent Your answer should be self-contained, which means that you can repeat things you already wrote as well as make new points. 

\subsubsection*{Candidate Summaries: 1/5}

Here is a possible summary:

\begin{quote}
``I believe that abortion should be legal and accessible because women have the right to make decisions about their own bodies. Access to safe abortions is crucial for protecting women's health and well-being. Society should support comprehensive sex education and contraception to reduce the need for abortions.''
\end{quote}

\noindent How well does this summary capture your viewpoint on abortion?\\

\noindent Choices: very poorly, poorly, moderately, well, very well, excellently, exceptionally \\

\noindent Explain which parts you \textbf{agree} with, which parts you \textbf{disagree} with, and what would need to be \textbf{added or made more concrete} to fully represent your viewpoint. \\

\noindent (Please write \textbf{two} or more sentences.)

\subsubsection*{Candidate Summaries: 2/5}

Here is a possible summary:

\begin{quote}
``I think abortion should be a personal decision between a woman and her doctor, without government interference. Each situation is unique, and women should have the autonomy to make the best choice for themselves and their families. Society should ensure that all women have access to affordable healthcare, including reproductive services.''
\end{quote}

\noindent How well does this summary capture your viewpoint on abortion?\\

\noindent Choices: very poorly, poorly, moderately, well, very well, excellently, exceptionally \\

\noindent Explain which parts you \textbf{agree} with, which parts you \textbf{disagree} with, and what would need to be \textbf{added or made more concrete} to fully represent your viewpoint. \\

\noindent (Please write \textbf{two} or more sentences.)

\subsubsection*{Candidate Summaries: 3/5}

Here is a possible summary:

\begin{quote}
``I believe that abortion should be allowed in the first trimester but restricted afterward unless there are exceptional circumstances. This policy respects a woman's right to choose while recognizing the increasing moral considerations as the pregnancy progresses. Society should invest in education and healthcare to prevent unwanted pregnancies and support women through their reproductive choices.''
\end{quote}

\noindent How well does this summary capture your viewpoint on abortion?\\

\noindent Choices: very poorly, poorly, moderately, well, very well, excellently, exceptionally \\

\noindent Explain which parts you \textbf{agree} with, which parts you \textbf{disagree} with, and what would need to be \textbf{added or made more concrete} to fully represent your viewpoint. \\

\noindent (Please write \textbf{two} or more sentences.)

\subsubsection*{Candidate Summaries: 4/5}

Here is a possible summary:

\begin{quote}
``I think abortion should be restricted to certain circumstances, such as cases of rape, incest, or when the mother's life is at risk. This approach balances the rights of the unborn with the needs of women facing difficult situations. Society should provide support for women who carry their pregnancies to term, including healthcare and financial assistance.''
\end{quote}

\noindent How well does this summary capture your viewpoint on abortion?\\

\noindent Choices: very poorly, poorly, moderately, well, very well, excellently, exceptionally \\

\noindent Explain which parts you \textbf{agree} with, which parts you \textbf{disagree} with, and what would need to be \textbf{added or made more concrete} to fully represent your viewpoint. \\

\noindent (Please write \textbf{two} or more sentences.)

\subsubsection*{Candidate Summaries: 5/5}

Here is a possible summary:

\begin{quote}
    ``I believe that abortion should be illegal because it involves taking a human life, which I consider morally wrong. Society should focus on providing resources and support for pregnant women to encourage them to choose life. Adoption should be promoted as a viable alternative to abortion.''
\end{quote}

\noindent How well does this summary capture your viewpoint on abortion?\\

\noindent Choices: very poorly, poorly, moderately, well, very well, excellently, exceptionally \\

\noindent Explain which parts you \textbf{agree} with, which parts you \textbf{disagree} with, and what would need to be \textbf{added or made more concrete} to fully represent your viewpoint. \\

\noindent (Please write \textbf{two} or more sentences.)

\subsubsection*{Thank You}

Thank you for completing the survey!\\

\noindent Click below to register your survey completion with Prolific and return to their website.

\subsection{Validation Survey}

\subsubsection*{Informed Consent}

We are a team of university researchers. We want to understand in detail what people like you think about social questions, in this case about abortion. We want to study the diversity of people’s beliefs, and whether algorithmic tools can help summarize and analyze these differing points of view.\\

\noindent\textit{What will I need to do and how long will the study last?}
We will ask you \textbf{10 questions}.
We expect that you will be in this research study for less than an hour. \\

\noindent\textit{Compensation: }
Your pay will \textbf{not} depend on your opinions: all good-faith responses will get fully compensated, so please just write what you think.
You may not use ChatGPT or other AI tools to write or edit your responses.\\

\noindent\textit{Who will see your responses?}
The data you provide will be immediately anonymized. We may later on publish this anonymous data. By continuing this survey, you agree to this use of your responses.\\

\subsubsection*{Instructions}

In this survey, we want to hear about your opinions on \textbf{abortion}.

To help researchers and policymakers understand public opinion on abortion, we condensed the responses of previous participants into \textbf{10 short summaries}.
We will now present each summary to you, and your task is to rate to what extent it would capture \textbf{your} point of view.

An ideal summary should

\begin{enumerate}
\item perfectly capture all your thoughts and feelings on abortion,
\item include the reasoning behind your opinion, and
\item be concrete and actionable.
\end{enumerate}

\subsubsection*{Candidate Summaries: 1/10}

Here is a possible summary:

\begin{quote}
    ``I believe abortion should be legal because it is a woman's right to choose what happens with her body. It is important for maintaining reproductive freedom and ensuring safe and accessible healthcare options. Women need to make decisions that are best for them and their families without government interference.''
\end{quote}

\noindent How well does this summary capture your viewpoint on abortion?\\

\noindent Choices: very poorly, poorly, moderately, well, very well, excellently, exceptionally \\

\noindent Explain which parts you \textbf{agree} with, which parts you \textbf{disagree} with, and what would need to be \textbf{added or made more concrete} to fully represent your viewpoint.\\

\noindent Please write \textbf{two} or more sentences.

\subsubsection*{Candidate Summaries: 2/10}

Here is a possible summary:

\begin{quote}
    ``In my opinion, abortion should be illegal except in cases where the mother's life is at risk or in instances of rape or incest. This stance protects unborn children while allowing necessary exceptions. It's about balancing moral concerns with compassion for those in difficult situations.''
\end{quote}

\noindent How well does this summary capture your viewpoint on abortion?\\

\noindent Choices: very poorly, poorly, moderately, well, very well, excellently, exceptionally \\

\noindent Explain which parts you \textbf{agree} with, which parts you \textbf{disagree} with, and what would need to be \textbf{added or made more concrete} to fully represent your viewpoint.\\

\noindent Please write \textbf{two} or more sentences.

\subsubsection*{Candidate Summaries: 3/10}

Here is a possible summary:

\begin{quote}
    ``Abortion should be legal up to a certain point in the pregnancy, such as the first or second trimester. After that, I think it needs to be restricted unless there are severe medical reasons. Society needs to find a middle ground that respects both a woman's right to choose and the potential life of the fetus.''
\end{quote}

\noindent How well does this summary capture your viewpoint on abortion?\\

\noindent Choices: very poorly, poorly, moderately, well, very well, excellently, exceptionally \\

\noindent Explain which parts you \textbf{agree} with, which parts you \textbf{disagree} with, and what would need to be \textbf{added or made more concrete} to fully represent your viewpoint.\\

\noindent Please write \textbf{two} or more sentences.

\subsubsection*{Candidate Summaries: 4/10}

Here is a possible summary:

\begin{quote}
    ``I think abortion should be legal and accessible at all stages of pregnancy. Women need to have control over their reproductive health for various personal and medical reasons. Restrictions can force women into unsafe or undesirable circumstances.''
\end{quote}

\noindent How well does this summary capture your viewpoint on abortion?\\

\noindent Choices: very poorly, poorly, moderately, well, very well, excellently, exceptionally \\

\noindent Explain which parts you \textbf{agree} with, which parts you \textbf{disagree} with, and what would need to be \textbf{added or made more concrete} to fully represent your viewpoint.\\

\noindent Please write \textbf{two} or more sentences.

\subsubsection*{Candidate Summaries: 5/10}

Here is a possible summary:

\begin{quote}
    ``Abortion should be illegal in all cases, as I believe it is morally wrong and equivalent to ending a potential life. Society should look for alternative solutions such as adoption and support for pregnant women. We need to value and protect all human life, starting from conception.''
\end{quote}

\noindent How well does this summary capture your viewpoint on abortion?\\

\noindent Choices: very poorly, poorly, moderately, well, very well, excellently, exceptionally \\

\noindent Explain which parts you \textbf{agree} with, which parts you \textbf{disagree} with, and what would need to be \textbf{added or made more concrete} to fully represent your viewpoint.\\

\noindent Please write \textbf{two} or more sentences.

\subsubsection*{Candidate Summaries: 6/10}

Here is a possible summary:

\begin{quote}
    ``I believe that abortion should be legal and that a woman has the right to choose what she does with her body. It is crucial for the decision to have an abortion to remain a healthcare decision between a woman and her doctor, free from political interference. Society should respect and uphold these rights to ensure women's autonomy and wellbeing.''
\end{quote}

\noindent How well does this summary capture your viewpoint on abortion?\\

\noindent Choices: very poorly, poorly, moderately, well, very well, excellently, exceptionally \\

\noindent Explain which parts you \textbf{agree} with, which parts you \textbf{disagree} with, and what would need to be \textbf{added or made more concrete} to fully represent your viewpoint.\\

\noindent Please write \textbf{two} or more sentences.

\subsubsection*{Candidate Summaries: 7/10}

Here is a possible summary:

\begin{quote}
    ``I believe that abortion should be legal and accessible to everyone, as it is essential healthcare. It is critical for women to have control over their own bodies and make decisions about abortion without government interference. Safe and legal access to abortion ensures that women can make the best choices for their health and well-being.''
\end{quote}

\noindent How well does this summary capture your viewpoint on abortion?\\

\noindent Choices: very poorly, poorly, moderately, well, very well, excellently, exceptionally \\

\noindent Explain which parts you \textbf{agree} with, which parts you \textbf{disagree} with, and what would need to be \textbf{added or made more concrete} to fully represent your viewpoint.\\

\noindent Please write \textbf{two} or more sentences.

\subsubsection*{Candidate Summaries: 8/10}

Here is a possible summary:

\begin{quote}
    ``I believe that abortion should be illegal in most cases and only allowed under specific circumstances. It is morally wrong and equates to taking a life, as I believe in the sanctity of human life starting from conception. People should take responsibility for their actions and should never use abortion as a form of birth control.''
\end{quote}

\noindent How well does this summary capture your viewpoint on abortion?\\

\noindent Choices: very poorly, poorly, moderately, well, very well, excellently, exceptionally \\

\noindent Explain which parts you \textbf{agree} with, which parts you \textbf{disagree} with, and what would need to be \textbf{added or made more concrete} to fully represent your viewpoint.\\

\noindent Please write \textbf{two} or more sentences.

\subsubsection*{Candidate Summaries: 9/10}

Here is a possible summary:

\begin{quote}
    ``I believe that abortion should be legal and that women should have the right to make decisions about their own bodies. The decision to have an abortion should be a private matter, free from societal judgment, as it's a crucial aspect of women's health and safety. Society should ensure that abortions are safe, legal, and accessible, respecting the individual's autonomy and privacy in making such decisions.''
\end{quote}

\noindent How well does this summary capture your viewpoint on abortion?\\

\noindent Choices: very poorly, poorly, moderately, well, very well, excellently, exceptionally \\

\noindent Explain which parts you \textbf{agree} with, which parts you \textbf{disagree} with, and what would need to be \textbf{added or made more concrete} to fully represent your viewpoint.\\

\noindent Please write \textbf{two} or more sentences.

\subsubsection*{Candidate Summaries: 10/10}

Here is a possible summary:

\begin{quote}
    ``I believe that abortion should not be used as a form of birth control. However, it should be legally permissible in cases where the mother's life is in danger, and in instances of rape or incest within an early timeframe. Furthermore, there should be robust educational efforts to prevent unwanted pregnancies and ensure people are informed about the complexities and consequences of abortion.''
\end{quote}

\noindent How well does this summary capture your viewpoint on abortion?\\

\noindent Choices: very poorly, poorly, moderately, well, very well, excellently, exceptionally \\

\noindent Explain which parts you \textbf{agree} with, which parts you \textbf{disagree} with, and what would need to be \textbf{added or made more concrete} to fully represent your viewpoint.\\

\noindent Please write \textbf{two} or more sentences.

\subsubsection*{Thank You}

Thank you for completing the survey!\\

\noindent Click below to register your survey completion with Prolific and return to their website.

\end{document}